\documentclass[journal]{IEEEtran}

\usepackage{amsmath,amssymb,amsfonts}
\usepackage{amsthm}
\usepackage{array}
\usepackage{microtype}
\usepackage{algorithm}
\usepackage{algpseudocode}
\usepackage{booktabs}
\usepackage{float}
\floatstyle{ruled}
\restylefloat{algorithm}
\usepackage{cite}
\usepackage{graphicx}
\usepackage{multirow}
\usepackage[caption=false,font=footnotesize]{subfig}
\usepackage{textcomp}
\usepackage{stfloats}
\usepackage{tabularx}
\usepackage{url}
\usepackage[table]{xcolor}

\theoremstyle{definition}
\newtheorem{definition}{Definition}
\newtheorem{assumption}{Assumption}
\theoremstyle{remark}
\newtheorem{remark}{Remark}
\theoremstyle{plain}
\newtheorem{proposition}{Proposition}
\newtheorem{lemma}{Lemma}
\newtheorem{corollary}{Corollary}
\newcolumntype{Y}{>{\raggedright\arraybackslash}X}
\newcolumntype{C}{>{\centering\arraybackslash}X}

\newcommand{\TableStyle}{%
    \scriptsize
    \renewcommand{\arraystretch}{1.12}%
    \setlength{\tabcolsep}{4.0pt}%
}
\newcommand{\SingleTableWidth}{0.94\columnwidth}
\newcommand{\DoubleTableWidth}{0.985\textwidth}
\newcommand{\TableSidePad}{3pt}

\newcommand{\TableNote}[1]{%
    \par\vspace{2pt}%
    {\noindent\footnotesize #1\par}%
    \vspace{1pt}%
}

\graphicspath{{fig/}{./}}

\begin{document}

\title{Channel Knowledge Map Reconstruction From Sparse Measurements via Pilot-Anchored Layout-Conditioned Fourier Refinement}

\author{Zhonghao~Jiu, Fan~Meng, Yongming~Huang,~\IEEEmembership{Fellow,~IEEE}, Hang~Zhan, Zening~Liu, Xiaohu~You,~\IEEEmembership{Fellow,~IEEE}
    \thanks{This work was supported in part by the National Natural Science Foundation of China under Grant Nos. 62225107 and 62201394, the Mobile Information Networks National Science and Technology Major Project under Grant 2025ZD1305000, the Natural Science Foundation on Frontier Leading Technology Basic Research Project of Jiangsu under Grant BK20222001, and the Fundamental Research Funds for the Central Universities under Grant 2242022k60002. (Corresponding author: Yongming Huang)}
    \thanks{Z. Jiu, Y. Huang, and X. You are with the School of Information Science and Engineering, Southeast University, Nanjing 210096, China. F. Meng, Y. Huang, H. Zhan, Z. Liu, and X. You are also with Purple Mountain Laboratories, Nanjing 211111, China (e-mail: 230228222@seu.edu.cn; mengfan@pmlabs.com.cn; huangym@seu.edu.cn; zhanxing@pmlabs.com.cn; liuzening@pmlabs.com.cn; xhyu@seu.edu.cn).}
}

\markboth{submitted to IEEE Transactions on Wireless Communications}%
{Jiu \MakeLowercase{\textit{et al.}}: CKM Reconstruction From Sparse Measurements}

\maketitle
\begin{abstract}

Channel knowledge maps (CKMs) enable environment-aware wireless systems by providing location-specific channel knowledge, but long-term environmental variations, such as construction, traffic redistribution, and foliage changes, require periodic map refresh. In practice, channel measurements are often sparse and irregular, while environmental knowledge may be limited to coarse layout or topology descriptors. This paper studies CKM reconstruction from sparse measurements. We show that reconstruction pipelines that apply local aggregation or spectral operators directly to a zero-filled pilot grid can entangle the sampling mask with the channel field, allowing structural priors to act on mask-induced distortions before the measurements define a supported radio field. To address this issue, we propose Anchor-CKM, a measurement-first, knowledge-aided reconstruction framework. Anchor-CKM first uses support-aware partial convolutions to construct a pilot-supported representation, and then performs layout-conditioned dual-path Fourier refinement followed by coordinate-based heteroscedastic prediction of the CKM mean and per-location predictive variance. Experiments on transmitter-disjoint DeepMIMO scenarios cover missing ratios from $0.3$ to $0.95$, including stringent $5\%$--$10\%$ pilot-coverage settings. In explicit-layout outdoor scenarios, Anchor-CKM reduces received-power root-mean-square error (RMSE) by $0.79$--$1.33$~dB relative to the strongest reproduced baseline, while ablations identify pilot-support stabilization as the largest contributor and layout conditioning as beneficial for line-of-sight/non-line-of-sight (LOS/NLOS) boundary fidelity.

\end{abstract}


\section{Introduction}
\label{sec:intro}

\IEEEPARstart{S}{ixth}-generation (6G) wireless networks are evolving from connectivity-centric systems toward infrastructures characterized by ubiquitous intelligence and sustainability. Beyond higher data rates, lower latency, and broader coverage, 6G is expected to integrate communication, sensing, artificial intelligence (AI), and digital twins to support real-time and controllable network optimization~\cite{you2025ai6g,li2026pemnet}. This shift requires the network to anticipate propagation conditions before each link is explicitly sounded, so that coverage management, handover, scheduling, and beamforming can exploit predicted channel conditions instead of relying solely on instantaneous per-link feedback.

Channel knowledge maps (CKMs)~\cite{zeng2024ckm_tutorial,lim2021beam_tracking,du2026ckm_dualdomain} provide a compact form of location-specific channel knowledge for this purpose. A CKM records large-scale propagation attributes over a geographic area, including received power, dominant angles of arrival, delay spread, and blockage statistics, and can be queried at candidate receiver positions. Unlike instantaneous channel state information (CSI), a CKM captures slowly varying, location-indexed propagation knowledge that can be reused across network decisions and digital-twin-assisted optimization. Once available, CKMs can support predictive beam management, proactive handover, coverage-hole identification, cell association, and network planning~\cite{zeng2024ckm_tutorial,survey2025radio}. Their practical value therefore depends not only on initial construction accuracy, but also on whether they can be maintained as the environment changes.

A CKM cannot be built once and used indefinitely. New construction, seasonal foliage, traffic-pattern changes, and deployment updates alter the effective radio environment, while each transmitter, carrier, or propagation state may require refreshed observations. Exhaustive drive tests, aerial surveys, dense pilot sweeps, and repeated high-fidelity simulations can track such changes, but their labor, spectrum, and coordination costs grow with the deployment area. The 3GPP minimization-of-drive-tests (MDT) initiative alleviates this burden by leveraging user-equipment reports~\cite{3gpp37320, hapsari2012mdt}; however, these opportunistic measurements are confined to service areas and inherit the spatial bias of user traffic.

Consequently, practical CKM refresh must reconstruct or update the map from sparse and irregular spatial measurements rather than idealized dense surveys. We model this setting by discretizing the valid probing region into spatial cells onto which scheduled pilots, mobile-user feedback, and crowd-sourced reports are aggregated, with a binary sampling mask recording which cells are probed. The estimator also has access to quasi-static environmental knowledge, such as coarse layout or topology descriptors, which evolves more slowly than the pilot measurements. The central question is how to exploit this slow prior without allowing it to overrule the limited current measurements.

\begin{figure}[t]
	\centering
	\includegraphics[width=\columnwidth]{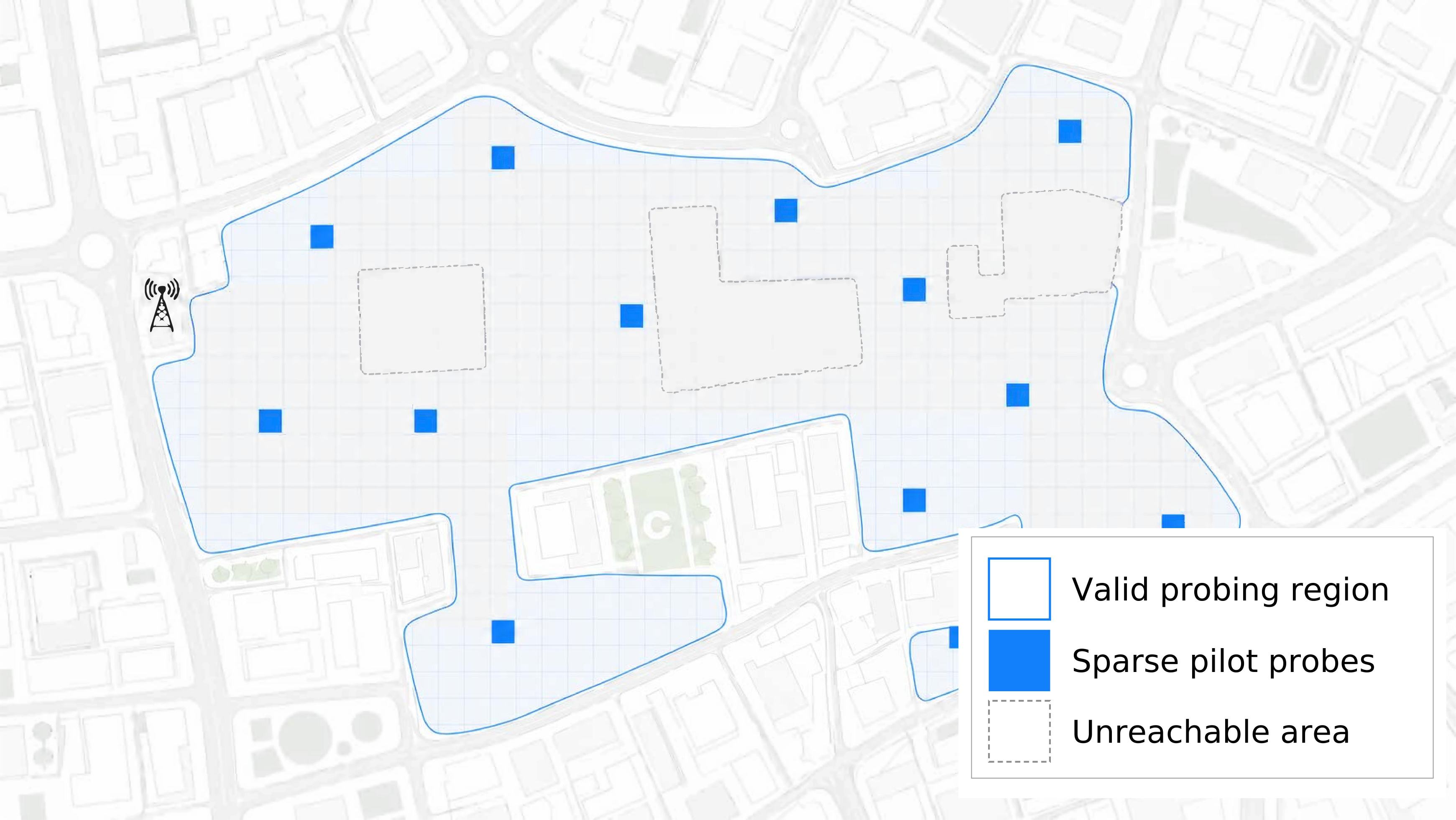}
	\caption{CKM refresh with sparse probing. Blue cells denote pilot probes; gray/dashed regions are excluded from the probing domain. Zeros at unprobed cells indicate absence of measurement, not physical channel values.}
	\label{fig:sparse_measurement_concept}
\end{figure}

\subsection{Related Work}
\label{sec:related}

Existing CKM and radio-map reconstruction methods primarily differ in the priors used to compensate for missing measurements. They broadly fall into four categories: classical spatial reconstruction, learned grid and graph completion, physics-aware operator and implicit field modeling, and reliability-aware reconstruction. We review each category with emphasis on how measurement evidence and structural priors are combined under sparse probing support.

\textit{Classical spatial reconstruction.} Geostatistical and model-based methods treat measured cells as samples from a spatial random field. Kriging~\cite{cressie1990origins}, inverse-distance weighting~\cite{shepard1968two}, low-rank matrix completion~\cite{candes2009exact}, and compressed sensing~\cite{candes2006robust} can be effective when the sampled support is sufficiently dense and the field is smooth or low-rank over the region of interest. When the measurement support becomes sparse, however, many unobserved cells may be far from any pilot, while local line-of-sight/non-line-of-sight (LOS/NLOS) transitions can violate stationarity or global-structure assumptions. The key risk is that the assumed regularizer, which is fixed before any probe is collected, dominates predictions wherever nearby measurements are scarce.

\textit{Learned grid and graph completion.} Learning-based methods replace explicit covariance models with learned spatial priors. RadioUNet~\cite{levie2021radiounet} and PMNet~\cite{lee2023pmnet} fuse a scene representation with a pilot-observation grid, whereas graph neural networks (GNNs) and attention/message-passing variants~\cite{he2023gnn_wireless,corso2020pna,hou2022graphmae,jiu2026rieif} model non-Euclidean measurement or geometry relations; RadioGAT~\cite{li2024radiogat} is a representative CKM-oriented graph model. These architectures are valuable for encoding layout and nonlocal relations. Under low pilot coverage, however, their earliest aggregation layers often operate on a masked grid or a grid-derived neighborhood graph, so sampling-induced distortions from the zero-filled grid can enter before a propagation field is recovered.

\textit{Physics-aware, operator, and implicit field modeling.} A stronger line of work constrains the reconstructed field itself. Diffusion-based radio-map synthesis~\cite{radiodiff2025} uses a learned generative prior under measurement constraints. Physics-informed neural networks (PINNs)~\cite{raissi2019pinn,liu2025pinn_gnn} encode propagation laws as soft constraints, although multi-objective PINN training can be unstable when the governing relation is empirical~\cite{wang2021ntk_pinn,xiao2024neural_ode}. Fourier neural operators (FNOs)~\cite{li2021fno,tran2023ffno,li2023gino} exploit low-mode structure in large-scale fields and are beginning to appear in wireless extrapolation and multiple-input multiple-output (MIMO) studies~\cite{xiao2025fno_mimo,gao2026channel_extrapolation}. Implicit representations such as sinusoidal representation networks (SIREN) and neural radiance fields (NeRF)~\cite{sitzmann2020siren,mildenhall2020nerf} provide continuous prediction over receiver coordinates. These tools strengthen the field prior, but they do not by themselves determine when that prior should act relative to sparse measurement evidence.

\textit{Reliability-aware reconstruction.} CKM refresh from sparse measurements also benefits from uncertainty or reliability reporting. Active-learning radio-map reconstruction uses uncertainty to select informative measurement locations when campaigns are costly~\cite{polyzos2024bayesian}. Such reporting steers where to measure next, not how priors and current measurements are ordered within the reconstructor. In this paper, predictive variance is used more narrowly as a deployment-side reliability indicator near coverage holes, LOS/NLOS boundaries, and scene edges, while the mean estimate remains constrained by the current sparse measurements.

In summary, prior work provides useful reconstruction mechanisms, but CKM refresh from sparse measurements exposes a common ordering vulnerability. Classical estimators can lean heavily on stationarity; learned models can process mask-induced distortions in early layers; and field-level priors can shape the estimate before the measurement support is stabilized.

\subsection{Contributions}
\label{sec:contributions}
The above comparison of sparse-reconstruction methods exposes a representation-level limitation in CKM refresh from sparse measurements. The issue is not simply whether the environmental or spectral prior is strong enough, but whether it is applied before the sparse measurements have formed a reliable radio-field representation. As illustrated in Fig.~\ref{fig:sparse_measurement_concept}, a zero-filled pilot grid contains physical measurements only at the probed cells, while zeros elsewhere denote absence of measurement rather than true channel values. Local aggregation or spectral mixing on this raw grid can therefore entangle the sampling mask with the channel field before a pilot-supported representation has formed. In the spectral view, spatial masking convolves the channel spectrum with a broadband mask spectrum, producing mask-induced spectral distortion. Under uniform random probing with missing ratio $\rho$, naive spectral debiasing inflates the operator-input variance by a factor of $\rho/(1-\rho)$, which grows rapidly as the probing budget shrinks. This factor is not a universal penalty for every estimator; it isolates one mechanism by which early aggregation or spectral processing on the raw masked grid can weaken the tie to current measurements.

We therefore develop Anchor-CKM, a measurement-first, knowledge-aided framework for CKM refresh from sparse measurements. Anchor-CKM adopts a staged reconstruction order: sparse pilots first form a supported radio-field representation, and only then are layout-conditioned spectral and local priors used for propagation-aware refinement. This ordering keeps the slow environmental prior useful without allowing it to replace the current sparse measurements. The main contributions are summarized as follows.
\begin{itemize}
	\item \textbf{CKM refresh from sparse measurements with mask-distortion analysis.}
	We formulate measurement-limited CKM refresh from sparse spatial measurements and coarse layout or topology side information. We show that a zero-filled pilot grid can be mask-dominated rather than a partially observed radio field: the sampling mask disrupts local neighborhoods and aliases with the channel spectrum. A second-moment analysis under random probing quantifies this effect through the variance-inflation factor $\rho/(1-\rho)$, providing an analytical rationale for stabilizing pilot support before imposing spectral or environmental priors.
	
	\item \textbf{A pilot-anchored layout-conditioned Fourier framework.}
	We propose Anchor-CKM, which implements this ordering through three components: a partial-convolution initializer that forms a pilot-supported field, a layout-conditioned dual-path Fourier operator for propagation-aware refinement, and a coordinate-based heteroscedastic prediction head that outputs both the reconstructed CKM and per-location predictive variance. This design separates fast pilot-driven refresh from slow environmental encoding.
	
	\item \textbf{Deployment-oriented evaluation.}
	We evaluate Anchor-CKM on transmitter-disjoint DeepMIMO benchmarks with shared sparse-probing masks and explicit layout or topology priors, across missing ratios from $0.3$ to $0.95$, including stringent $5\%$--$10\%$ pilot-coverage settings. At $\rho=0.9$, Anchor-CKM reduces received-power root-mean-square error (RMSE) by $1.33$ and $0.79$~dB on the explicit-layout O1-28 and O1-60 scenarios relative to the strongest reproduced baseline, markedly improves low-coverage detection and best-cell selection, and obtains the lowest RMSE on the indoor I1-2.5 portability test under a weaker indirect prior. Ablations and robustness checks identify pilot-support stabilization as the largest observed contributor, with layout conditioning improving LOS/NLOS boundary fidelity and predictive variance providing a secondary reliability cue near coverage-hole and noisy-pilot regions.
\end{itemize}

The remainder of this paper is organized as follows. Section~\ref{sec:system_model} defines CKM refresh from sparse measurements, including the valid probing domain, spatial mask, task objective, and direct-reconstruction limitation. Section~\ref{sec:method} presents Anchor-CKM and the analytical motivation for its staged ordering. Section~\ref{sec:experiments} reports the experimental evaluation, and Section~\ref{sec:conclusion} concludes the paper.

\section{System Model and Task Formulation}
\label{sec:system_model}

We first specify CKM refresh from sparse measurements independently of any reconstruction architecture: the CKM field, valid probing domain and mask, missing ratio, estimator information set, and expected-risk objective. We then characterize the limitation of direct reconstruction in this regime. Table~\ref{tab:symbols} lists the main symbols used in the formulation and method description.

\subsection{Channel Modeling}
\label{sec:channel_model}

Let $\Omega\subset\mathbb{R}^{2}$ denote a two-dimensional (2D) deployment region of interest discretized into a regular $H\times W$ grid $\mathcal{G}=(V,E)$ with $N=HW=|V|$ nodes. Each node $i\in V$ corresponds to a spatial cell centered at $\mathbf{x}_{i}\in\Omega$; the transmitter (TX) is located at $\mathbf{x}_{\mathrm{TX}}$.

The following large-scale model serves as a structural abstraction for CKM behavior; the ray-tracing experiments do not assume it as their data-generating process. Under this abstraction, the received power at cell $i$ is governed by distance-dependent path loss and shadow fading. Using the standard log-distance model~\cite{rappaport2002wireless},
\begin{equation}
    P_{\mathrm{LS}}(\mathbf{x}_{i})
    = P_{\mathrm{TX}} - 10\,n\log_{10}\!\left(\frac{\|\mathbf{x}_{i}-\mathbf{x}_{\mathrm{TX}}\|}{r_{0}}\right) + X_{\sigma}(\mathbf{x}_{i})
    \quad [\mathrm{dB}],
    \label{eq:path_loss}
\end{equation}
where $n$ is the path-loss exponent, $r_{0}$ is a reference distance, and $X_{\sigma}(\mathbf{x}_{i})\sim\mathcal{N}(0,\sigma_{\mathrm{SF}}^{2})$ is the shadow-fading term. Shadow fading is spatially correlated and is classically modeled by an exponentially decaying kernel in Gudmundson's formulation~\cite{gudmundson1991shadow}, so the large-scale field varies slowly over regions in which the propagation condition remains homogeneous.

\begin{table}[!t]
\centering
\caption{Main notation.}
\label{tab:symbols}
\begin{minipage}{0.98\columnwidth}
\centering
\scriptsize
\renewcommand{\arraystretch}{1.08}
\setlength{\tabcolsep}{2.4pt}
\begin{tabularx}{\linewidth}{@{\hspace{2pt}}l Y@{\hspace{2pt}}}
\toprule
\textbf{Symbol} & \textbf{Meaning} \\
\midrule
\multicolumn{2}{l}{\textbf{Task and channel model}} \\
$\Omega,\mathcal{G},N,H,W$ & region, grid graph, cell count, and grid size ($N=HW$) \\
$\mathbf{x}_{i},\mathbf{x}_{\mathrm{TX}},\Omega_{\mathrm{bldg}}$ & cell/transmitter coordinates and obstacle region \\
$P_{\mathrm{TX}},P_{\mathrm{LS}},\bar{P},\widetilde{P}$ & transmit power, large-scale power, trend, and residual field \\
$\operatorname{LOS}(\cdot),\boldsymbol{\theta}(\mathbf{x}),d_c$ & LOS indicator, local propagation parameters, and shadowing correlation distance \\
$\mathbf{p}_{i},\mathbf{P},F$ & per-cell attributes, CKM field, and feature dimension \\
$\Omega_{\mathrm{valid}},\Omega_o,\Omega_u$ & valid, observed, and unobserved valid sets \\
$v_i,m_i,\mathbf{M},\mathbf{P}_{\mathrm{obs}}$ & validity indicator, observation indicator, broadcast mask, and masked CKM \\
$\rho,p,\mathcal{S},q_{\rho}$ & missing ratio, probing probability, sampled support, and mask law \\
\addlinespace[2pt]
\multicolumn{2}{l}{\textbf{Estimator, architecture, and loss}} \\
$\Pi,\mathcal{G}_{\mathrm{env}},f_{\theta}$ & environmental prior, site-layout/topology graph, and estimator \\
$\mathcal{D}_{\mathrm{CKM}},\mathcal{R}_{\theta},d_{\mathrm{CKM}}$ & data distribution, task risk, and discrepancy metric \\
$\hat{\boldsymbol{\mu}},\hat{\boldsymbol{\sigma}}^{2},\hat{\mathbf{P}}$ & predicted mean, predictive variance, and reconstructed CKM \\
$\mathcal{F},\mathcal{F}_{u},\widehat{\mathbf{Z}},\boldsymbol{\Delta}$ & Fourier transforms, debiased spectrum, and centered mask \\
$\mathbf{X}_0,\mathbf{X}^{(l)},\boldsymbol{\Phi}$ & pilot-supported initial field, operator latent, and propagated latent field \\
$L_F,C,(K_h,K_w),R$ & conditional Fourier-operator depth, latent width, retained modes, and factorization rank \\
$\boldsymbol{\gamma}_{l},\boldsymbol{\beta}_{l},\mathbf{Z}_{\mathrm{spec}}^{(l)},\mathbf{Z}_{\mathrm{spat}}^{(l)}$ & modulation fields and spectral/local updates \\
$\mathbf{W}_{\mathrm{spec}}^{(l)},\boldsymbol{\Theta}^{(l)}$ & spectral weights and local convolutional branch \\
$\mathbf{S},\mathbf{h}_i,\mathbf{e}_{ij},b_{ij}$ & graph-to-grid feature map, node embedding, edge feature, and attention bias \\
$\mathcal{V}_{\mathrm{qry}},\mathbf{x}_{\mathrm{q}},\boldsymbol{\psi},L_{\mathrm{rff}}$ & requested receiver set/coordinate, Fourier coordinate features, and their count \\
$s_{i,f},Q_{\mathrm{off}},\eta,r_{b,f,j}$ & clipped log variance, off-grid sample count, interpolated validity, and off-grid residual \\
$\mathcal{L}_{\mathrm{total}},\mathcal{L}_{\mathrm{off}},w_{\mathrm{obs}},\lambda_{\mathrm{unc}}$ & total/off-grid losses, observed-cell weight, and variance regularizer \\
\bottomrule
\end{tabularx}
\TableNote{$\mathcal{F}$ is the non-unitary 2D discrete Fourier transform (DFT) in~\eqref{eq:mask_leakage}; $\mathcal{F}_u$ is the unitary DFT in Lemma~\ref{lem:mask_distortion} and Corollary~\ref{cor:naive_debias}.}
\end{minipage}
\end{table}

LOS/NLOS transitions introduce sharp local deviations wherever obstacles block the direct path. Let $\Omega_{\mathrm{bldg}}$ denote the union of obstacle footprints. With the LOS indicator $\operatorname{LOS}(\mathbf{x}_{i},\mathbf{x}_{\mathrm{TX}})=\mathbf{1}(\operatorname{seg}(\mathbf{x}_{i},\mathbf{x}_{\mathrm{TX}})\cap\Omega_{\mathrm{bldg}}=\varnothing)$, the effective propagation parameters $\boldsymbol{\theta}(\mathbf{x})=[\,n(\mathbf{x}),\,\sigma_{\mathrm{SF}}(\mathbf{x})\,]^{\top}$ become spatially varying, not globally stationary~\cite{rappaport2002wireless}, producing sharp local nonstationarity at street corners and building edges that the reconstructor should preserve.

The same model also motivates an operator-level bias toward low spatial modes. Away from transmitter singularities and abrupt blockage boundaries, the long-range component of a propagation field is dominated by relatively low spatial frequencies under standard parabolic-wave approximations~\cite{levy2000parabolic,apaydin2010ssf}. This observation motivates low-mode propagation, while the local LOS/NLOS deviations motivate a separate local correction path. These physical abstractions guide Anchor-CKM's architecture; the model does not solve the governing partial differential equation (PDE) explicitly.

Beyond received power, a CKM may collect multiple channel attributes such as dominant angle of arrival (AoA), delay spread, or LOS probability. For feature dimension $F$, the per-cell attributes are stacked into a multi-attribute field $\mathbf{P}=\{\mathbf{p}_{i}\}_{i=1}^{N}\in\mathbb{R}^{N\times F}$, which is the target of the reconstruction task. The experiments instantiate this definition with a received-power-only setting and a three-target setting containing received power, AoA azimuth, and AoA elevation.

\subsection{Pilot-Based Sparse Spatial Measurement Model}
\label{sec:probing}

Every CKM cell value must ultimately be inferred from physical probes such as pilot snapshots at active users, drive-test campaigns, or opportunistic crowd-sourced measurement reports~\cite{zeng2024ckm_tutorial,survey2025radio}. Spatial probing density is therefore limited by pilot and signaling budgets, survey costs, and the spatial nonuniformity of measurement locations, so the input to any CKM reconstructor is spatially incomplete by construction.

Not every grid cell is physically meaningful: building interiors, out-of-coverage regions, or locations structurally unreachable for user devices must be excluded from the reconstruction target. We introduce a per-cell validity indicator $v_{i}\in\{0,1\}$ and define the valid domain
\begin{equation}
    \Omega_{\mathrm{valid}} = \{\, i \in V : v_{i}=1 \,\},
    \label{eq:valid_domain}
\end{equation}
so that probing is restricted to $\Omega_{\mathrm{valid}}$. The probing event at a valid cell is captured by a binary observation indicator $m_{i}\in\{0,1\}$, with $m_{i}=1$ iff cell $i$ has been probed. The masked observation is written compactly as
\begin{equation}
    \mathbf{P}_{\mathrm{obs}} = \mathbf{M}\odot\mathbf{P},
    \label{eq:obs_model}
\end{equation}
where $\mathbf{M}$ is the feature-broadcast form of $\{m_{i}\}$ and $\odot$ denotes elementwise multiplication. Zeros in $\mathbf{P}_{\mathrm{obs}}$ denote missing entries only through the pair $(\mathbf{P}_{\mathrm{obs}},\mathbf{M})$; they are not interpreted as physical channel values. In the reported experiments, the observation mask is cell-level and is broadcast across the reconstructed attributes; feature-dependent probing masks can be incorporated by replacing $\mathbf{M}$ with an attribute-specific mask. The observed and unobserved valid-index sets are
\begin{equation}
    \Omega_{o} = \{\, i : v_{i}=1,\, m_{i}=1 \,\},\quad
    \Omega_{u} = \{\, i : v_{i}=1,\, m_{i}=0 \,\},
    \label{eq:index_sets}
\end{equation}
and are used both to construct the sparse-measurement input and to evaluate reconstruction error on unobserved valid cells.

Sparsity carries several distinct meanings in the wireless literature: sparse angular or delay channel support, sparse training data, low-rank channel structure, or sparse spatial probing locations. To avoid ambiguity, we fix the usage through the following definition and the accompanying probing assumption.

\begin{definition}[Spatial-probing sparsity and missing ratio]
\label{def:sparsity}
Throughout this paper, spatial-probing sparsity refers exclusively to the fraction of cells in the valid domain $\Omega_{\mathrm{valid}}$ at which no measurement is available; it does not describe the angular/delay support of the physical channel, nor the size of the training set. The corresponding missing ratio is
\begin{equation}
    \rho \;=\; \frac{|\Omega_{u}|}{|\Omega_{\mathrm{valid}}|},
    \qquad
    |\Omega_{o}| \;=\; (1-\rho)\,|\Omega_{\mathrm{valid}}|.
    \label{eq:sparsity_def}
\end{equation}
\end{definition}
The formulation applies to any missing ratio; the evaluation in Section~\ref{sec:experiments} additionally emphasizes stringent low-coverage settings $\rho\in[0.9,0.95]$, i.e., $5\%$--$10\%$ pilot coverage over the valid cells.

\begin{assumption}[Uniform random spatial probing]
\label{asm:sampling}
For a missing ratio $\rho$ prescribed by Definition~\ref{def:sparsity}, the set of probed cells is drawn uniformly at random from the valid domain. Specifically, a subset $\mathcal{S}\subset\Omega_{\mathrm{valid}}$ of cardinality $n_o=|\Omega_{\mathrm{valid}}|-\lfloor\rho|\Omega_{\mathrm{valid}}|\rfloor$ is sampled without replacement, and the observation indicator is set as $m_{i}=\mathbf{1}(i\in\mathcal{S})$. The marginal probing probability is therefore
\begin{equation}
    \Pr(m_{i}=1 \mid v_{i}=1) = \frac{n_o}{|\Omega_{\mathrm{valid}}|},
    \qquad \forall\,i\in\Omega_{\mathrm{valid}},
    \label{eq:prob_sampling}
\end{equation}
which equals $1-\rho$ up to integer rounding and is subject to the hard cardinality constraint $|\Omega_{o}|=n_o$.
\end{assumption}

\begin{remark}[Random and grid-uniform probing]\label{rem:random_sampling}
Uniform random probing serves as a controlled model of irregular support: it produces spatially heterogeneous gaps similar to opportunistic acquisition, while keeping the missing ratio a single adjustable scalar. It does not reproduce the traffic-induced spatial bias of opportunistic reports, which we leave to future work on trajectory-driven sampling. Periodic decimation, by contrast, creates highly structured holes whose mask spectrum differs from that of irregular field measurements.
\end{remark}

\subsection{Problem Formulation}
\label{sec:problem_formulation}

Together, these elements define the task-level reconstruction problem. In addition to the sparse measurements $\mathbf{P}_{\mathrm{obs}}$ and the observation mask $\mathbf{M}$, the estimator has access to a quasi-static environmental prior $\Pi$ that encodes coarse graph-structured knowledge about the propagation context. Depending on the deployment, $\Pi$ encodes either explicit site-layout assets such as obstacle boundaries, transmitter locations, visibility relations, and zone adjacency, or weaker topological cues derived from auxiliary signatures when full map assets are unavailable. The key property of $\Pi$ is its slow evolution relative to instantaneous channel samples. It can complement sparse measurements without being refreshed at the pilot rate, while the reconstruction path itself is recomputed for each pilot-refresh input.

The parametric reconstruction model is written as
\begin{equation}
    (\hat{\boldsymbol{\mu}},\hat{\boldsymbol{\sigma}}^{2})
    = f_{\theta}(\mathbf{P}_{\mathrm{obs}},\mathbf{M},\Pi),
    \label{eq:generic_predictor}
\end{equation}
where $\hat{\boldsymbol{\mu}}\in\mathbb{R}^{N\times F}$ is the predicted CKM mean and $\hat{\boldsymbol{\sigma}}^{2}\in\mathbb{R}^{N\times F}$ is an optional spatially varying predictive-variance output used for deployment-side assessment. We use $\hat{\mathbf{P}}$ only when referring to the reconstructed CKM field as a whole. The joint dependence of $\hat{\boldsymbol{\mu}}$ on both $(\mathbf{P}_{\mathrm{obs}},\mathbf{M})$ and $\Pi$ distinguishes the task from classical interpolation, which does not use environmental side information, and from map synthesis models that are not explicitly anchored to the current sparse measurements.

For training, let $\mathcal{D}_{\mathrm{CKM}}$ denote the distribution of CKM/prior pairs $(\mathbf{P},\Pi)$, and let $q_{\rho}(\mathbf{M}\mid\Omega_{\mathrm{valid}})$ denote the probing-mask law induced by Assumption~\ref{asm:sampling}. For a sampled mask, the unprobed valid-cell set is written as $\Omega_u(\mathbf{M})$ to emphasize that the supervised evaluation locations are determined by the same pilot support that forms the input. The resulting task-level risk is
\begin{equation}
\begin{aligned}
    \theta^{\star}
    = \arg\min_{\theta}\;&
    \mathbb{E}_{(\mathbf{P},\Pi)\sim\mathcal{D}_{\mathrm{CKM}}}
    \mathbb{E}_{\mathbf{M}\sim q_{\rho}(\cdot\mid\Omega_{\mathrm{valid}})}
    \big[\mathcal{R}_{\theta}(\mathbf{P},\mathbf{M},\Pi)\big], \\
    (\hat{\boldsymbol{\mu}}_{\theta},\hat{\boldsymbol{\sigma}}^{2}_{\theta})
    &= f_{\theta}(\mathbf{M}\odot\mathbf{P},\mathbf{M},\Pi),
\end{aligned}
\label{eq:generic_objective}
\end{equation}
where the per-mask reconstruction risk is
\begin{equation}
    \mathcal{R}_{\theta}
    = \frac{1}{|\Omega_u(\mathbf{M})|}
      \sum_{i\in\Omega_u(\mathbf{M})}
      d_{\mathrm{CKM}}\!\left(
        \mathbf{p}_{i},
        \hat{\boldsymbol{\mu}}_{\theta,i},
        \hat{\boldsymbol{\sigma}}^{2}_{\theta,i}
      \right).
    \label{eq:masked_risk}
\end{equation}
Here $d_{\mathrm{CKM}}$ is a generic task-level discrepancy evaluated only on unobserved valid cells, so the objective measures reconstruction instead of probe memorization. Possible instantiations of $d_{\mathrm{CKM}}$ include squared-error, angular, probabilistic-scoring, and communication-oriented surrogates. The Anchor-CKM implementation uses a heteroscedastic likelihood with light observed-cell consistency and off-grid consistency terms, while the formulation itself remains independent of that particular loss choice.

\subsection{Limitations of Direct CKM Reconstruction}
\label{sec:classical_methods}

This formulation clarifies why direct reconstruction becomes weakly tied to measurements as the probing support shrinks. Geostatistical interpolation methods such as ordinary kriging~\cite{cressie1990origins} and distance-weighted schemes such as inverse-distance weighting~\cite{shepard1968two} estimate an unobserved location as
\begin{equation}
    \hat{p}(\mathbf{x}_{\star})
    = \sum_{i\in\Omega_{o}} \lambda_{i}(\mathbf{x}_{\star})\,p(\mathbf{x}_{i}),
    \label{eq:kriging}
\end{equation}
with weights $\lambda_{i}(\mathbf{x}_{\star})$ obtained from a covariance model or distance kernel; in stringent low-coverage cases such as $5\%$--$10\%$ pilot coverage, most predictions are supported by only a few distant samples and local LOS/NLOS transitions violate the stationarity assumptions that make these estimators reliable. Compressed sensing and low-rank recovery~\cite{candes2006robust,candes2009exact} share the same difficulty in a different form: at high missing ratios, the reconstruction is driven more by the regularizer than by the available measurements, and these methods do not directly incorporate environmental geometry.

The same support issue also affects learned direct reconstruction when the first layers operate on the zero-filled observation grid before a pilot-supported representation has been formed. Although the mask $\mathbf{M}$ identifies missing entries, local aggregation and spectral transforms are still applied to neighborhoods and frequency statistics shaped by the sampling pattern, so the estimator must separate measured channel values from missing-entry zeros while simultaneously extrapolating the field; at high missing ratios these two roles become difficult to disentangle. The bottleneck is therefore the order in which information enters the model: the sparse pilot support should first be converted into an operator-ready field, with environmental and spectral refinement applied only afterward. Section~\ref{sec:method} develops this ordering and its analytical motivation.

\section{Anchor-CKM Method}
\label{sec:method}

Anchor-CKM uses a staged, measurement-first reconstruction. Fig.~\ref{fig:arch} summarizes the two time scales: cached layout-conditioned modulation on the slow scale, and pilot-driven densification, spectral propagation, and continuous location-wise prediction at each refresh. The corresponding modules are a partial-convolution densify-to-space (PConv-DtS) initializer, a conditional Fourier neural operator (C-FNO) modulated by feature-wise linear modulation (FiLM)~\cite{perez2018film}, and a residual multi-layer perceptron (MLP) prediction head with predictive variance for deployment-side assessment.

\begin{figure*}[t]
    \centering
    \includegraphics[width=\textwidth]{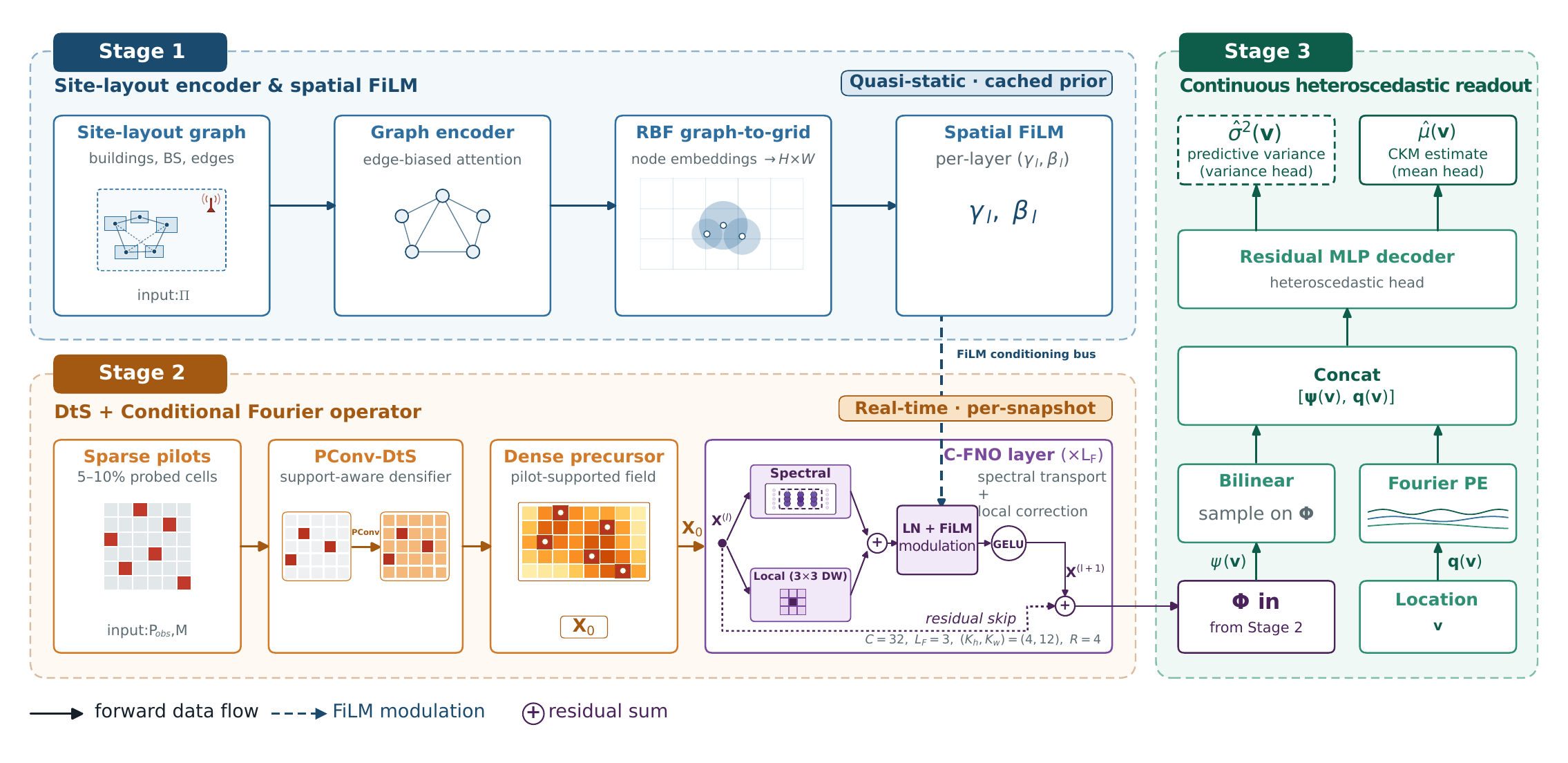}
    \caption{Anchor-CKM architecture. Stage~1 caches layout-conditioned FiLM fields; Stage~2 applies PConv-DtS densification and FiLM-modulated C-FNO propagation per snapshot; Stage~3 predicts continuous mean and predictive variance. Dashed arrows denote cached conditioning.}
    \label{fig:arch}
\end{figure*}

Let $\mathbf{X}_0$ denote the pilot-supported initial field, $\boldsymbol{\Phi}$ the propagated latent field, $\mathbf{x}_{\mathrm{q}}$ a continuous receiver coordinate, and $\boldsymbol{\Gamma}=\{(\boldsymbol{\gamma}_{l},\boldsymbol{\beta}_{l})\}_{l=1}^{L_F}$ the cached collection of spatial FiLM fields. The Anchor-CKM forward pass factorizes as
\begin{equation}
\begin{aligned}
    \boldsymbol{\Gamma}
    &= f_{\mathrm{layout}}(\mathcal{G}_{\mathrm{env}}),
    && \text{(Stage~1, cached)} \\
    \mathbf{X}_0
    &= f_{\mathrm{PConv}}(\mathbf{P}_{\mathrm{obs}}, \mathbf{M}),
    && \text{(Stage~2, DtS)} \\
    \boldsymbol{\Phi}
    &= f_{\mathrm{C\mbox{-}FNO}}\!\left(
        \mathbf{X}_0 \mid \boldsymbol{\Gamma}
    \right),
    && \text{(Stage~2, C-FNO)} \\
    \big(\hat{\boldsymbol{\mu}}(\mathbf{x}_{\mathrm{q}}), \hat{\boldsymbol{\sigma}}^{2}(\mathbf{x}_{\mathrm{q}})\big)
    &= f_{\mathrm{dec}}(\boldsymbol{\Phi}, \mathbf{x}_{\mathrm{q}}).
    && \text{(Stage~3)}
\end{aligned}
\label{eq:Anchor-CKM_pipeline}
\end{equation}
Stage~1 can be computed before a pilot snapshot arrives, but the cached fields affect the radio map only after DtS forms $\mathbf{X}_0$.

\subsection{Analytical Motivation}
\label{sec:physics_design}

The masked pilot grid is a poor direct input to a spectral operator, but the underlying large-scale field becomes suitable for low-mode spectral propagation once pilot support has been stabilized. All expectations in this argument are taken over the random probing mask conditioned on a fixed CKM field $\mathbf{P}$, and the DFT is applied featurewise when $F>1$. Since $\mathbf{P}_{\mathrm{obs}}=\mathbf{M}\odot\mathbf{P}$, the 2D discrete Fourier transform (DFT) of the pilot-observation grid obeys
\begin{equation}
    \mathcal{F}[\mathbf{P}_{\mathrm{obs}}]
    = \mathcal{F}[\mathbf{P}\odot\mathbf{M}]
    = \frac{1}{N}\,\mathcal{F}[\mathbf{P}]\ast\mathcal{F}[\mathbf{M}],
    \label{eq:mask_leakage}
\end{equation}
so the sampling pattern convolves the true spectrum with a broadband mask spectrum before any operator sees the field. For the second-moment analysis, the valid-cell field and mask are embedded on the $H\times W$ grid, with invalid cells excluded from the energy sums. This keeps the DFT notation rectangular while preserving the missing-ratio dependence over $\Omega_{\mathrm{valid}}$. We use a unitary-normalized DFT; under the convention of~\eqref{eq:mask_leakage}, the same $\rho$-dependence holds up to a constant factor.

\begin{lemma}[Mask-induced spectral distortion under uniform probing]
\label{lem:mask_distortion}
Under Assumption~\ref{asm:sampling} with $p=1-\rho$ and $\mathcal{F}_{u}$ the unitary 2D DFT under this embedding, $\mathbb{E}[\mathcal{F}_{u}[\mathbf{P}_{\mathrm{obs}}]]=p\,\mathcal{F}_{u}[\mathbf{P}]$ and
\begin{equation}
    \mathbb{E}\!\left[
        \left\|
            \mathcal{F}_{u}[\mathbf{P}_{\mathrm{obs}}]
            - p\,\mathcal{F}_{u}[\mathbf{P}]
        \right\|_F^2
    \right]
    = \rho(1-\rho)\,\|\mathbf{P}\|_F^2.
\end{equation}
\end{lemma}

\begin{proof}
Let $\boldsymbol{\Delta}=\mathbf{M}-p\mathbf{1}$. Then $\mathbb{E}[\boldsymbol{\Delta}]=\mathbf{0}$, giving the mean result. The residual equals $\mathcal{F}_{u}[\boldsymbol{\Delta}\odot\mathbf{P}]$; by Parseval, its squared norm is $\sum_{i\in\Omega_{\mathrm{valid}}}\Delta_i^2\|\mathbf{p}_i\|_2^2$. Taking expectations with $\mathbb{E}[\Delta_i^2]=\rho(1-\rho)$ gives the stated expression.
\end{proof}

\begin{corollary}[Variance blow-up under naive spectral debiasing]
\label{cor:naive_debias}
With the same notation as Lemma~\ref{lem:mask_distortion}, define the naively debiased masked spectrum $\widehat{\mathbf{Z}} \triangleq p^{-1}\mathcal{F}_{u}[\mathbf{P}_{\mathrm{obs}}]$. Then $\mathbb{E}[\widehat{\mathbf{Z}}]=\mathcal{F}_{u}[\mathbf{P}]$ and
\begin{equation}
    \mathbb{E}\!\left[
        \left\|
            \widehat{\mathbf{Z}}-\mathcal{F}_{u}[\mathbf{P}]
        \right\|_F^2
    \right]
    = \frac{1-p}{p}\,\|\mathbf{P}\|_F^2
    = \frac{\rho}{1-\rho}\,\|\mathbf{P}\|_F^2.
\end{equation}
Therefore, as $\rho\rightarrow 1$, a spectral input formed by naively debiasing the raw masked grid incurs an increasingly large variance penalty even before model mismatch is considered.
\end{corollary}

\begin{proof}
This follows by substituting $\widehat{\mathbf{Z}}=p^{-1}\mathcal{F}_{u}[\mathbf{P}_{\mathrm{obs}}]$ and applying Lemma~\ref{lem:mask_distortion}, which scales the residual energy by $p^{-2}$.
\end{proof}
The variance factor is already $9$ at $\rho=0.9$ and $19$ at $\rho=0.95$, which argues against spectral operations on the raw masked grid.

Once pilot support has been stabilized, however, the smooth large-scale component is well suited to low-mode propagation. After removing the deterministic distance trend in~\eqref{eq:path_loss}, the residual shadowing field has the following low-pass spectral form.

\begin{proposition}[Spectral compactness of the large-scale component]
\label{prop:spectral_compactness}
Let
\begin{equation}
    P_{\mathrm{LS}}(\mathbf{x})
    = \bar{P}(\mathbf{x}) + X_{\sigma}(\mathbf{x}),
\end{equation}
where $\bar{P}(\mathbf{x}) = P_{\mathrm{TX}} - 10n\log_{10}\!\big(\|\mathbf{x}-\mathbf{x}_{\mathrm{TX}}\|/r_0\big)$ is the deterministic log-distance trend and the shadow-fading term has covariance $K(r)=\sigma_{\mathrm{SF}}^2 e^{-r/d_c}$ with correlation distance $d_c$. On any bounded tile that excludes the transmitter singularity, the fluctuation field $\widetilde{P}(\mathbf{x}) \triangleq P_{\mathrm{LS}}(\mathbf{x})-\bar{P}(\mathbf{x})$ has isotropic power spectral density
\begin{equation}
    S_{\widetilde{P}}(\nu)
    = \frac{2\pi \sigma_{\mathrm{SF}}^2 d_c^2}{\left(1+d_c^2 \nu^2\right)^{3/2}},
    \qquad \nu=\|\mathbf{k}\|_2,
\end{equation}
and therefore satisfies $S_{\widetilde{P}}(\nu)\le C(1+d_c^2\nu^2)^{-1}$ for some constant $C>0$. Hence most modeled shadowing energy is concentrated below the characteristic spatial frequency scale $1/d_c$, while the high-frequency tail decays polynomially.
\end{proposition}

The expression follows from the Hankel transform of the isotropic exponential covariance. Proposition~\ref{prop:spectral_compactness} justifies restricting spectral propagation to low modes. Local nonstationarity remains: LOS/NLOS parameters can change abruptly near buildings and street corners, so Anchor-CKM pairs the low-mode C-FNO branch with a local correction path modulated by layout-conditioned FiLM fields.

\subsection{Stage 1: Cached Site-Layout Encoder and Spatial FiLM}
\label{sec:layout_encoder}

Stage~1 operates on the cached time scale: it converts quasi-static environmental knowledge, represented as a graph, into per-layer spatial modulation fields $(\boldsymbol{\gamma}_{l},\boldsymbol{\beta}_{l})$ for the conditional Fourier operator. Because the prior changes far more slowly than pilot reports, this stage is computed once and reused across snapshots.

\subsubsection{Environmental Prior and Layout Graph}
\label{sec:Anchor-CKM_prior}

Anchor-CKM realizes the prior $\Pi$ as a site-layout graph $\mathcal{G}_{\mathrm{env}}=(V_{\mathrm{env}},E_{\mathrm{env}})$, which serves as quasi-static environmental graph knowledge. With explicit layouts, nodes represent the transmitter, building corners, and retained-building centroids; edges encode perimeters, corner-to-building membership, transmitter visibility, and local geometric context. Without explicit layouts, the same prior representation is constructed from indirect topology zones, such as strongest non-target base station (BS) components on I1-2.5, providing cached side information without target CKM values or current pilots.

\subsubsection{Graph Encoder and Spatial FiLM Projection}
\label{sec:envgnn}

The layout encoder extracts spatially varying propagation cues from $\mathcal{G}_{\mathrm{env}}$. Node interactions use edge-biased attention:
\begin{equation}
    \begin{aligned}
        \alpha_{ij}
        &= \operatorname*{softmax}_{j}\!\left(
            \frac{(\mathbf{Q}\mathbf{h}_i)^\top(\mathbf{K}\mathbf{h}_j)}{\sqrt{d}}
            + b_{ij}
        \right), \\
        b_{ij}
        &= \mathbf{w}_e^\top \mathbf{e}_{ij},
    \end{aligned}
    \label{eq:attention}
\end{equation}
where the additive bias $b_{ij}$ injects inter-node geometry directly into the attention logits~\cite{vaswani2017attention,ying2021graphormer}. The resulting node embeddings are projected to the $H\times W$ grid by Gaussian radial-basis-function (RBF) interpolation:
\begin{equation}
\begin{aligned}
    \kappa_i(h,w)
    &= \exp\!\left(-\frac{\|\mathbf{x}_{hw}-\mathbf{x}_i\|^2}{2\sigma_{\mathrm{rbf}}^2}\right), \\
    \omega_i(h,w)
    &= \frac{\kappa_i(h,w)}{\sum_{j=1}^{|V_{\mathrm{env}}|}\kappa_j(h,w)}, \\
    \mathbf{S}(h,w)
    &= \sum_{i=1}^{|V_{\mathrm{env}}|}\omega_i(h,w)\,\mathbf{h}_i,
\end{aligned}
\label{eq:rbf}
\end{equation}
where $\sigma_{\mathrm{rbf}}=\exp(\tilde{\sigma})$ is a positive learnable bandwidth. The resulting spatial feature map is then mapped to per-layer FiLM fields:
\begin{equation}
    \big(\boldsymbol{\gamma}_l(h,w),\, \boldsymbol{\beta}_l(h,w)\big)
    = \operatorname{Conv}_{1\times 1}\!\big(\mathbf{S}(h,w)\big),
    \quad l = 1, \ldots, L_F.
    \label{eq:spatial_film}
\end{equation}
This construction is cell-dependent: nearby geometry dominates each cell's modulation, so different regions receive different operator responses. The entire output $\{(\boldsymbol{\gamma}_l,\boldsymbol{\beta}_l)\}_{l=1}^{L_F}$ is computed once and cached; only the pilot-dependent propagation path is recomputed on each refresh.

\subsection{Stage 2: Pilot-Supported Conditional Fourier Neural Operator}
\label{sec:cfno}

Stage~2 operates on the fast per-snapshot scale and runs for each current sparse pilot grid. It contains a pilot-support initialization sub-module followed by a C-FNO propagation sub-module that uses the cached FiLM modulation $(\boldsymbol{\gamma}_l,\boldsymbol{\beta}_l)$.

\subsubsection{Pilot-Support Initialization via Densify-to-Space}
\label{sec:dts}

The first sub-module of Stage~2 turns scattered pilots into an operator-ready, pilot-supported field. Mask distortion arises before any operator propagation, so the support must be stabilized before spectral or local operations act on the latent field. We implement this stabilization using support-aware partial convolutions~\cite{liu2018pconv}, which update only with valid neighbors and grow the effective support with depth.

With $s_w = \lVert \mathbf{M}_w \rVert_1$ denoting the number of valid entries in the local mask patch, the PConv update at each spatial position is
\begin{equation}
    \mathbf{x}' =
    \begin{cases}
        \dfrac{k_{\mathrm{pc}}^2}{s_w}\,\mathbf{W}^\top(\mathbf{x} \odot \mathbf{M}_w) + \mathbf{b},
        & \text{if } s_w > 0, \\[2pt]
        \mathbf{0}, & \text{if } s_w = 0.
    \end{cases}
    \label{eq:pconv}
\end{equation}
The factor $k_{\mathrm{pc}}^2/s_w$ keeps the response roughly invariant to how much pilot support is available in the receptive field, and the updated mask $\mathbf{M}'_w = \mathbf{1}(s_w > 0)$ lets the effective support expand with depth. The three-layer stack used here produces an initial field $\mathbf{X}_0$ for C-FNO propagation without treating missing entries as physical channel values; a light observed-cell consistency term during training keeps $\mathbf{X}_0$ consistent with the probed cells.

\subsubsection{Conditional Fourier Neural Operator}
\label{sec:cfno_operator}

Once the DtS sub-module has stabilized the pilot support, Stage~2 applies the C-FNO with two complementary paths: a global spectral branch for smooth long-range propagation and a local spatial branch for blockage-boundary correction. The spatial FiLM fields $(\boldsymbol{\gamma}_l,\boldsymbol{\beta}_l)$ cached from Stage~1 modulate both paths so that the same operator layer can respond differently across LOS corridors and shadow regions.

Let $\mathbf{X}^{(0)}=\mathbf{X}_0$. For layer $l=1,\ldots,L_F$, the conditioned dual-path propagation evolves as
\begin{equation}
\begin{aligned}
    \mathbf{Z}_{\mathrm{spec}}^{(l)}
    &= \mathcal{F}^{-1}\!\left(\mathbf{W}_{\mathrm{spec}}^{(l)} \odot \mathcal{F}(\mathbf{X}^{(l-1)})\right), \\
    \mathbf{Z}_{\mathrm{spat}}^{(l)}
    &= \mathbf{X}^{(l-1)} \circledast \boldsymbol{\Theta}^{(l)}, \\
    \Delta\mathbf{X}^{(l)}
    &= \boldsymbol{\gamma}_{l} \odot \operatorname{LN}\!\left(
        \mathbf{Z}_{\mathrm{spec}}^{(l)} + \mathbf{Z}_{\mathrm{spat}}^{(l)}
    \right) + \boldsymbol{\beta}_{l}, \\
    \mathbf{X}^{(l)}
    &= \mathbf{X}^{(l-1)} + \sigma_{\mathrm{GELU}}\!\left(\Delta\mathbf{X}^{(l)}\right).
\end{aligned}
\label{eq:cfno_layer}
\end{equation}
Here $\operatorname{LN}$ denotes layer normalization, $\sigma_{\mathrm{GELU}}$ denotes the Gaussian error linear unit (GELU), $\mathcal{F}$ and $\mathcal{F}^{-1}$ are 2D DFTs, $\mathbf{W}_{\mathrm{spec}}^{(l)}$ acts only on the retained $K_h\times K_w$ low-frequency modes, and $\boldsymbol{\Theta}^{(l)}$ is a local convolutional correction branch. The cell-wise FiLM fields modulate the normalized sum of spectral and spatial responses. In the reported implementation, retained spectral weights are rank-$R$ factorized and the local branch is depthwise separable. After $L_F$ layers, $\boldsymbol{\Phi}=\mathbf{X}^{(L_F)}$.

\subsection{Stage 3: Continuous Location-Wise Prediction with Predictive Variance}
\label{sec:decoder}

To estimate the CKM at an arbitrary receiver location, the spatial coordinate $\mathbf{x}_{\mathrm{q}} = (x, y)$ is encoded via random Fourier features~\cite{tancik2020fourier}:
\begin{equation}
    \boldsymbol{\psi}(\mathbf{x}_{\mathrm{q}})
    = \big[\cos(2\pi\mathbf{B}\mathbf{x}_{\mathrm{q}})^{\top},\,
       \sin(2\pi\mathbf{B}\mathbf{x}_{\mathrm{q}})^{\top}\big]^{\top}
    \in \mathbb{R}^{2L_{\mathrm{rff}}},
    \label{eq:rff}
\end{equation}
where $\mathbf{B} \in \mathbb{R}^{L_{\mathrm{rff}} \times 2}$ is a fixed random matrix; we refer to $\boldsymbol{\psi}(\mathbf{x}_{\mathrm{q}})$ as Fourier coordinate features. Let $\mathbf{q}(\mathbf{x}_{\mathrm{q}}) = \operatorname{bilinear}(\boldsymbol{\Phi}, \mathbf{x}_{\mathrm{q}})$ denote the interpolated latent feature. The continuous prediction head is expressed as
\begin{equation}
\begin{aligned}
    \mathbf{z}(\mathbf{x}_{\mathrm{q}})
    &= \big[\boldsymbol{\psi}(\mathbf{x}_{\mathrm{q}})^{\top},\, \mathbf{q}(\mathbf{x}_{\mathrm{q}})^{\top}\big]^{\top}, \\
    \mathbf{h}(\mathbf{x}_{\mathrm{q}})
    &= \mathcal{H}_{\mathrm{res}}\!\left(
        \mathbf{W}_{\mathrm{in}}\mathbf{z}(\mathbf{x}_{\mathrm{q}}) + \mathbf{b}_{\mathrm{in}}
    \right), \\
    \hat{\mu}_f(\mathbf{x}_{\mathrm{q}})
    &= \mathbf{w}_{\mu,f}^{\top}\mathbf{h}(\mathbf{x}_{\mathrm{q}}) + b_{\mu,f}, \\
    \log\hat{\sigma}_f^2(\mathbf{x}_{\mathrm{q}})
    &= \mathbf{w}_{\sigma,f}^{\top}\mathbf{h}(\mathbf{x}_{\mathrm{q}}) + b_{\sigma,f}.
\end{aligned}
\label{eq:decoder}
\end{equation}
Here $\mathcal{H}_{\mathrm{res}}$ is a two-layer residual MLP, and the log-variance output is clipped to $[-9,4]$ for numerical stability (we use $\operatorname{clip}_{[a,b]}(x)\triangleq\max\{a,\min\{x,b\}\}$ throughout). The prediction head is continuous across arbitrary coordinates without per-scene fitting: the latent field carries the learned propagation representation, and the coordinate-based head remains lightweight.

To train the mean prediction and obtain a predictive-variance output, Anchor-CKM uses a heteroscedastic Gaussian regression loss on valid nodes, following the scoring-rule form commonly used for aleatoric predictive variance~\cite{kendall2018multi}. Let $\Omega_o$ and $\Omega_u$ denote the observed and unobserved valid-node sets. For feature $f$ and either set $S\in\{o,u\}$, the per-set heteroscedastic negative log-likelihood (NLL) is
\begin{equation}
    \mathcal{L}_f^{(S)}
    = \frac{1}{|\Omega_S|}
      \sum_{i \in \Omega_S}
      \left[
        \tfrac{1}{2} e^{-s_{i,f}} \big(y_{i,f} - \hat{\mu}_{i,f}\big)^2
        + \tfrac{1}{2} s_{i,f}
      \right],
    \label{eq:nll}
\end{equation}
where $s_{i,f} = \operatorname{clip}_{[-8,4]}(\log \hat{\sigma}_{i,f}^2)$ further tightens the log-variance clipping range in~\eqref{eq:decoder} for gradient stability of the heteroscedastic term, and $\hat{\mu}_{i,f}$ is the predicted mean for target $f$ at node $i$. The unobserved term $\mathcal{L}_f^{(u)}$ carries the main supervision, whereas $\mathcal{L}_f^{(o)}$ is kept deliberately light so that the prediction head remains consistent with probed cells without letting the probed set dominate learning. The reported models are trained with
\begin{equation}
    \begin{aligned}
        \mathcal{L}_{\mathrm{total}}
        &= \sum_f \big(
            \mathcal{L}_f^{(u)}
            + w_{\mathrm{obs}}\, \mathcal{L}_f^{(o)}
        \big)
        + \mathbf{1}_{\mathrm{phase2}}\, \mathcal{L}_{\mathrm{off}} \\
        &\quad + \lambda_{\mathrm{unc}}
        \sum_f \tfrac{1}{|\Omega_u|}
        \sum_{i \in \Omega_u}
        \big(\log \hat{\sigma}_{i,f}^2\big)^2,
    \end{aligned}
    \label{eq:total}
\end{equation}
where $\mathbf{1}_{\mathrm{phase2}}$ activates the off-grid term only in the second training phase, $w_{\mathrm{obs}}$ controls the observed-cell consistency penalty, and $\lambda_{\mathrm{unc}}$ bounds the log-variance head. During phase~2, each training tile draws $Q_{\mathrm{off}}$ continuous coordinates $\mathbf{x}_{\mathrm{q},b,j}$ in the normalized tile domain. Let $\mathcal{F}_{\mathrm{tar}}$ denote the set of reconstructed target attributes, $\mathcal{I}_{h}[\mathbf{P}_{f}]$ denote the bilinear interpolation of the dense training target, $\eta_{b,j}$ the interpolated valid-domain indicator, and $r_{b,f,j}=\hat{\mu}_{f}(\mathbf{x}_{\mathrm{q},b,j})-\mathcal{I}_{h}[\mathbf{P}_{f}](\mathbf{x}_{\mathrm{q},b,j})$. The off-grid term is
\begin{equation}
    \mathcal{L}_{\mathrm{off}}
    = \frac{1}{B|\mathcal{F}_{\mathrm{tar}}|}
    \sum_{b=1}^{B}\sum_{f\in\mathcal{F}_{\mathrm{tar}}}
    \frac{\sum_{j} \eta_{b,j}\, r_{b,f,j}^{2}}{\sum_{j} \eta_{b,j} + \epsilon},
    \label{eq:offgrid_loss}
\end{equation}
where the inner sums run over $j=1,\dots,Q_{\mathrm{off}}$. This term regularizes interpolation between grid nodes during supervised training only; it is not a PDE residual and is not used at inference, where Anchor-CKM receives only sparse pilots, the mask, and the cached prior.

The cached inference flow is listed in Algorithm~\ref{alg:Anchor-CKM}.

\begin{algorithm}[!t]
\caption{Measurement-first Anchor-CKM inference}
\label{alg:Anchor-CKM}
\footnotesize
\begin{algorithmic}[1]
\Require Sparse pilots $\mathbf{P}_{\mathrm{obs}}$, mask $\mathbf{M}$, prior $\Pi$, requested receiver set $\mathcal{V}_{\mathrm{qry}}$
\Ensure $\{\hat{\boldsymbol{\mu}}(\mathbf{x}_{\mathrm{q}}),\hat{\boldsymbol{\sigma}}^2(\mathbf{x}_{\mathrm{q}})\}_{\mathbf{x}_{\mathrm{q}}\in\mathcal{V}_{\mathrm{qry}}}$
\If{$\boldsymbol{\Gamma}$ is unavailable or $\Pi$ has changed}
    \State Build $\mathcal{G}_{\mathrm{env}}$ from $\Pi$, encode it, project node embeddings to the grid, and cache $\boldsymbol{\Gamma}=\{(\boldsymbol{\gamma}_l,\boldsymbol{\beta}_l)\}_{l=1}^{L_F}$.
\Else
    \State Reuse the cached $\boldsymbol{\Gamma}$.
\EndIf
\State Initialize $\mathbf{X}^{(0)} \leftarrow f_{\mathrm{PConv}}(\mathbf{P}_{\mathrm{obs}},\mathbf{M})$ by DtS.
\For{$l=1,\ldots,L_F$}
    \State Compute low-mode spectral propagation and local spatial correction from $\mathbf{X}^{(l-1)}$.
    \State Modulate the combined update with $(\boldsymbol{\gamma}_l,\boldsymbol{\beta}_l)$.
    \State $\mathbf{X}^{(l)} \leftarrow \mathrm{C\mbox{-}FNO}_l(\mathbf{X}^{(l-1)};\boldsymbol{\gamma}_l,\boldsymbol{\beta}_l)$.
\EndFor
\State $\boldsymbol{\Phi}\leftarrow\mathbf{X}^{(L_F)}$.
\For{each requested location $\mathbf{x}_{\mathrm{q}}\in\mathcal{V}_{\mathrm{qry}}$}
    \State Interpolate $\boldsymbol{\Phi}$ at $\mathbf{x}_{\mathrm{q}}$ and concatenate Fourier coordinate features.
    \State Map the feature vector to mean and predictive variance $(\hat{\boldsymbol{\mu}},\hat{\boldsymbol{\sigma}}^2)$.
\EndFor
\end{algorithmic}
\end{algorithm}

During training, the forward pass is optimized using the loss in~\eqref{eq:total}; during inference, the predicted mean and variance are returned at the requested receiver locations.

\subsection{Complexity Analysis}
\label{sec:complexity}

Let $N=HW$, $C$ be the latent width, $L_F$ the number of C-FNO layers, $R$ the spectral rank, and $(K_h,K_w)$ the retained modes. With cached layout conditioning, Stage~2 inference costs
\begin{equation}
\begin{aligned}
\mathcal{C}_{\mathrm{cached}}
    = \mathcal{O}\!\Big(
        &L_F\big(CN\log N + C^2N + R C K_hK_w\big) \\
        &+ C_{\mathrm{pc}}N
    \Big),
\end{aligned}
\label{eq:complexity_Anchor-CKM}
\end{equation}
where the terms correspond to fast Fourier transform (FFT) propagation, channel mixing, retained-mode spectral multiplication, and the fixed-depth PConv initializer. With fixed width, rank, retained modes, and PConv kernel size, the grid-size scaling is summarized by the compact spectral-operator term $\mathcal{O}(L_F C N\log N)$ used in Table~\ref{tab:complexity}. If the environmental prior changes, graph encoding and graph-to-grid FiLM projection add
\begin{equation}
\mathcal{C}_{\mathrm{uncached}}
    = \mathcal{C}_{\mathrm{cached}}
    + \mathcal{O}\!\left(|E_{\mathrm{env}}|d^2 + N |V_{\mathrm{env}}| d\right).
\end{equation}
On the reported O1 tiles, $N=32\times181$ and $|V_{\mathrm{env}}|\ll N$, so the grid-domain terms dominate cached inference. Table~\ref{tab:complexity} reports the cached footprint against reproduced single-pass baselines; diffusion estimators are excluded because their per-inference cost scales with reverse-sampling steps.

\begin{table*}[t]
\centering
\caption{Cached inference footprint on one $32{\times}181$ O1 tile.}
\label{tab:complexity}
\begin{minipage}{\DoubleTableWidth}
\centering
\TableStyle
\setlength{\tabcolsep}{4pt}
\begin{tabularx}{\linewidth}{@{\hspace{\TableSidePad}}l C C C Y@{\hspace{\TableSidePad}}}
\toprule
\textbf{Model} & \textbf{Parameters (M)} & \textbf{FLOPs (G)} & \textbf{Latency (ms)} & \textbf{Dominant cost} \\
\midrule
RadioUNet            & 0.93 & 0.44 & 3.42 & $\mathcal{O}(k^{2} C^{2} N)$ \\
PMNet                & 0.52 & 0.84 & 2.85 & $\mathcal{O}(k^{2} C^{2} N)$ \\
RadioGAT             & 0.13 & 1.53 & 5.67 & $\mathcal{O}(|E| C^{2} + |V| C^{2})$ \\
\rowcolor[gray]{0.95}
\textbf{Anchor-CKM (Ours, cached)} & \textbf{0.27} & \textbf{0.98} & \textbf{2.53} & $\mathcal{O}(L_F C N \log N)$ \\
\bottomrule
\end{tabularx}
\TableNote{Latency: A100, batch size~$1$, $100$ runs after $30$ warm-ups. Cached Anchor-CKM excludes Stage~1; use uncached cost only when the prior changes.}
\end{minipage}
\end{table*}

\section{Experimental Evaluation}
\label{sec:experiments}

The evaluation covers whole-field accuracy and downstream decision metrics, robustness to missing ratio and pilot noise, boundary behavior, predictive-variance assessment, and multi-target reconstruction.

\subsection{Experimental Setup}
\label{sec:setup}

\subsubsection{Scenarios and Splits}
We evaluate on three DeepMIMO ray-tracing scenarios~\cite{alkhateeb2019deepmimo}: O1-28, O1-60, and I1-2.5. The outdoor O1-28 and O1-60 scenes come with explicit layout and form the main matched-scenario benchmark; I1-2.5 lacks comparable explicit layout assets, so we treat it as a portability test under a weaker indirect prior. This indirect prior is built from receiver coordinates, reference BS locations, the identity of the strongest non-target BS at each cell, and the zone adjacency induced by that signature grid.

All three splits are transmitter-disjoint, so no transmitter index appears in more than one of the train, validation, and test partitions. On O1-28 and O1-60, we assign transmitters $t003$--$t010$, $t011$--$t012$, and $t013$--$t014$ to the three partitions, yielding $1360/340/340$ tiles of size $32{\times}181 = 5792$ grid nodes each. On I1-2.5, we assign $t001$--$t008$, $t009$, and $t010$--$t012$ in the same order, yielding $192/24/72$ tiles of size $32{\times}201 = 6432$ grid nodes each. At the main-result missing ratio $\rho = 0.9$, an O1 test tile contains, on average, $580$ observed and $5212$ unobserved valid nodes, while an I1 test tile contains $644$ observed and $5788$ unobserved nodes. The per-tile mask is generated by drawing $\lfloor \rho\,|\Omega_{\mathrm{valid}}| \rfloor$ unobserved cells uniformly at random within the valid domain and treating the remaining valid cells as probed, which matches Assumption~\ref{asm:sampling} up to integer rounding.

\subsubsection{Baselines}
We compare against low-rank matrix completion (LRMC)~\cite{candes2009exact}, RadioUNet~\cite{levie2021radiounet}, PMNet~\cite{lee2023pmnet}, RadioGAT~\cite{li2024radiogat}, and RadioDiff~\cite{radiodiff2025}. All baselines share Anchor-CKM's observed pilot values, binary masks, transmitter-disjoint split, unobserved-valid-cell evaluator, and validation-tuned training budget. Side information is supplied only through supported input channels or graph features, and no method sees dense target values at inference. Our RadioDiff reproduction is a grid-domain conditional denoising diffusion model with cosine noise scheduling and denoising diffusion implicit model (DDIM) sampling using observation replacement. It tests whether a heavier generative prior is competitive at this data scale and pilot density.

\subsubsection{Evaluation Metrics}
For $\rho=0.9$, we report RMSE, mean absolute error (MAE), standardized-domain normalized mean-square error ($\mathrm{NMSE}_{z}$), and the 90th-percentile absolute error (P90) on unobserved valid nodes; dB improvements are quoted as absolute reductions. P90 captures upper-tail service holes, while Regret@25 and low-coverage F1/recall evaluate candidate-cell selection and coverage-hole flagging, respectively. Multi-target experiments report per-target standardized RMSE and their mean.

\subsubsection{Training Details}
Training uses dense DeepMIMO tiles; inference uses only sparse pilots, the mask, and the environmental prior. We use a two-phase curriculum, AdamW with cosine annealing, batch size $16$, and early stopping. In phase~2, Anchor-CKM uses $Q_{\mathrm{off}}=1024$ continuous off-grid samples per training tile, with loss weights $w_{\mathrm{obs}}=0.1$ and $\lambda_{\mathrm{unc}}=0.05$. Anchor-CKM O1 results are three-seed means over $\{42,123,7\}$ (RMSE std $0.113/0.102$~dB on O1-28/O1-60); baselines use best validation-tuned checkpoints under the same split, mask, and evaluator. The predictive-variance and noisy-pilot assessment tables use fixed evaluation masks and, unless stated, the seed-42 Anchor-CKM checkpoint. All models use PyTorch on NVIDIA A100 hardware.

Unless otherwise stated, the reported Anchor-CKM runs use latent width $C=32$, $L_F=3$ C-FNO layers, retained Fourier modes $(K_h,K_w)=(4,12)$, rank-$4$ spectral factorization, and a three-layer PConv densifier with $3\times3$ kernels and dilations $(1,2,1)$. The layout encoder uses three graph-attention layers with four heads and hidden width $32$, and the location-wise prediction head uses $L_{\mathrm{rff}}=12$ random Fourier frequencies, a $128$-wide residual MLP, and clipped log variance.

\subsection{Main Performance Analysis}
\label{sec:main_results}

\begin{figure*}[t]
    \centering
    \includegraphics[width=0.88\textwidth]{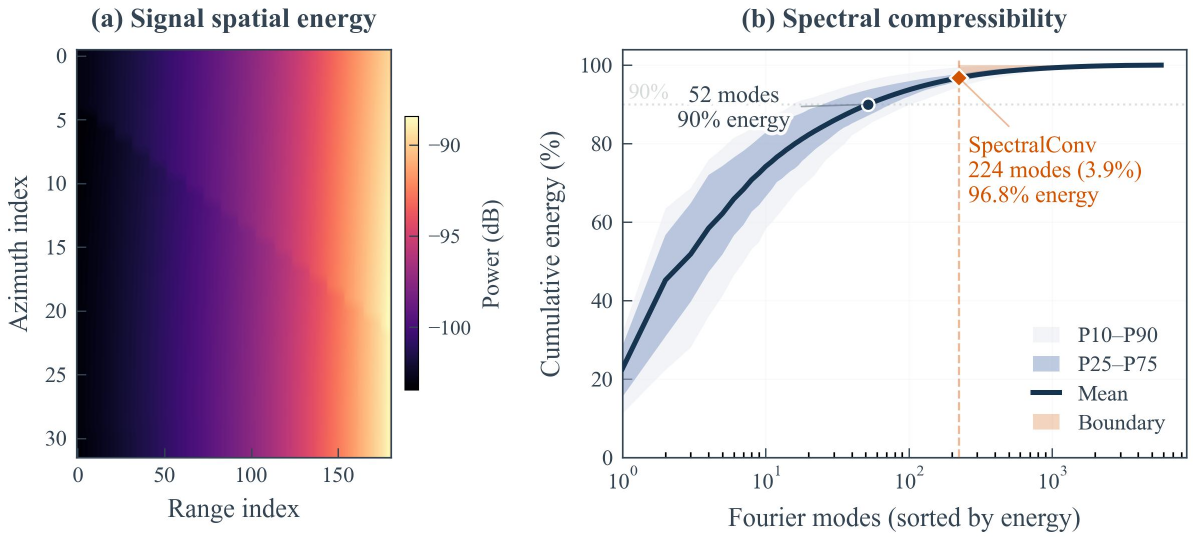}
    \caption{Fourier compressibility on O1-60: representative CKM and sorted-mode cumulative energy. Bands show tile percentiles; dashed line marks $90\%$ energy; orange marker is the C-FNO operating point.}
    \label{fig:fourier}
\end{figure*}

\subsubsection{Spectral Compressibility} Fig.~\ref{fig:fourier} shows that fewer than $52$ Fourier modes, under $1\%$ of coefficients, carry $90\%$ of O1-60 spatial-variation energy on average. The evaluated O1-60 fields are therefore not globally high-frequency; at $\rho=0.9$--$0.95$, the challenge is recovering low modes from mask-distorted observations while preserving local LOS/NLOS transitions. This supports the low-mode operator bias, with DtS still needed to stabilize pilot support.

\begin{table*}[!t]
\centering
\caption{Reconstruction accuracy at $\rho=0.9$.}
\label{tab:main}
\begin{minipage}{\textwidth}
\centering
\TableStyle
\setlength{\tabcolsep}{2.8pt}

\begin{tabularx}{\linewidth}{@{\hspace{\TableSidePad}}l *{8}{C}@{\hspace{\TableSidePad}}}
\toprule
\multirow{2}{*}{\textbf{Method}} & \multicolumn{4}{c}{\textbf{O1-28 (Outdoor)}} & \multicolumn{4}{c}{\textbf{O1-60 (Outdoor)}} \\
\cmidrule(lr){2-5} \cmidrule(lr){6-9}
& RMSE (dB) $\downarrow$ & MAE (dB) $\downarrow$ & $\mathrm{NMSE}_{z} \downarrow$ & P90 (dB) $\downarrow$ & RMSE (dB) $\downarrow$ & MAE (dB) $\downarrow$ & $\mathrm{NMSE}_{z} \downarrow$ & P90 (dB) $\downarrow$ \\
\midrule
LRMC & 7.288 & 5.797 & 1.000 & 12.216 & 9.329 & 7.593 & 1.000 & 15.117 \\
RadioUNet & 3.566 & 2.889 & 0.239 & 5.659 & 3.744 & 2.749 & 0.161 & 6.339 \\
RadioDiff & 13.662 & 10.984 & 3.515 & 22.540 & 16.931 & 13.548 & 3.295 & 28.068 \\
RadioGAT & 2.256 & 1.510 & 0.096 & 3.295 & 3.605 & 2.603 & 0.149 & 5.964 \\
PMNet & 2.045 & 1.557 & 0.079 & 3.371 & 1.781 & 1.229 & 0.036 & 2.997 \\
\rowcolor[gray]{0.95}
\textbf{Anchor-CKM (Ours)} & \textbf{0.720} & \textbf{0.476} & \textbf{0.010} & \textbf{1.015} & \textbf{0.988} & \textbf{0.663} & \textbf{0.011} & \textbf{1.377} \\
\bottomrule
\end{tabularx}

\vspace{4pt}

\begin{tabularx}{\linewidth}{@{\hspace{\TableSidePad}}l *{8}{C}@{\hspace{\TableSidePad}}}
\toprule
\multirow{2}{*}{\textbf{Method}} & \multicolumn{4}{c}{\textbf{I1-2.5 (Indoor)}} & \multicolumn{4}{c}{\textbf{Average (3 scenarios)}} \\
\cmidrule(lr){2-5} \cmidrule(lr){6-9}
& RMSE (dB) $\downarrow$ & MAE (dB) $\downarrow$ & $\mathrm{NMSE}_{z} \downarrow$ & P90 (dB) $\downarrow$ & RMSE (dB) $\downarrow$ & MAE (dB) $\downarrow$ & $\mathrm{NMSE}_{z} \downarrow$ & P90 (dB) $\downarrow$ \\
\midrule
LRMC & 0.486 & 0.393 & 1.000 & 0.852 & 5.701 & 4.594 & 1.000 & 9.395 \\
RadioUNet & 0.176 & 0.127 & 0.132 & 0.313 & 2.496 & 1.922 & 0.177 & 4.104 \\
RadioDiff & 1.307 & 1.039 & 7.216 & 2.161 & 10.633 & 8.524 & 4.675 & 17.590 \\
RadioGAT & 0.207 & 0.178 & 0.180 & 0.316 & 2.022 & 1.430 & 0.142 & 3.192 \\
PMNet & 0.308 & 0.275 & 0.401 & 0.450 & 1.378 & 1.021 & 0.172 & 2.273 \\
\rowcolor[gray]{0.95}
\textbf{Anchor-CKM (Ours)} & \textbf{0.165} & \textbf{0.124} & \textbf{0.115} & \textbf{0.260} & \textbf{0.624} & \textbf{0.421} & \textbf{0.046} & \textbf{0.884} \\
\bottomrule
\end{tabularx}
\TableNote{O1 Anchor-CKM: three-seed means (RMSE std $0.113$/$0.102$~dB); baselines: validation-selected single runs under the same split, masks, and evaluator. I1-2.5 uses an indirect prior; Average is the unweighted scenario mean for ranking only.}
\end{minipage}
\end{table*}

\subsubsection{Whole-Field Accuracy} Table~\ref{tab:main} reports the main comparison at $\rho=0.9$. O1-28 and O1-60 are the primary explicit-layout benchmarks; I1-2.5 probes portability under a weaker indirect prior.

Relative to PMNet, the strongest reproduced baseline, Anchor-CKM reduces RMSE from $2.045$ to $0.720$~dB on O1-28 and from $1.781$ to $0.988$~dB on O1-60. P90 drops from $3.371$ to $1.015$~dB on O1-28 and from $2.997$ to $1.377$~dB on O1-60. The P90 gains are important because upper-tail errors correspond to localized service holes, so the improvement is not confined to easy open-area cells. These gains are consistent with the intended ordering: geometry conditions a pilot-supported representation instead of substituting for the measurement constraint. RadioDiff underperforms under the same transmitter-disjoint split and budget, suggesting that this reproduction is mismatched to the present pilot density and data scale; diffusion priors may remain useful in denser or larger radio-map settings.

On I1-2.5, Anchor-CKM uses only the indirect prior yet attains the best RMSE and $\mathrm{NMSE}_{z}$. I1-2.5 has a different target scale and a smoother field, so the unweighted average is only a ranking summary. The narrower margin over RadioUNet ($0.165$ versus $0.176$~dB) should therefore be read as portability under a weaker prior, not as a gain comparable to the outdoor margins or as evidence that all baselines remain measurement-anchored under sparse measurements.

\subsubsection{Communication-Oriented Decision Metrics} Map accuracy matters when it supports downstream decisions. On the O1 test tiles and using the fixed seed-42 checkpoints, we evaluate two decision-oriented metrics: best-cell regret and low-coverage-region detection. Regret@25 measures how far the best ground-truth cell captured by the predicted top-$25$ candidate set falls below the true strongest unobserved cell. Low coverage is defined scenario-wise by the lower $10$th percentile of the unobserved-node received-power distribution, so the threshold adapts to the O1-28/O1-60 scales instead of importing a common dBm cutoff. For this detection task, we report the F1 score and recall.

\begin{table}[t]
\centering
\caption{Best-cell regret and low-coverage detection at $\rho=0.9$.}
\label{tab:decision_metrics}
\begin{minipage}{\SingleTableWidth}
\centering
\TableStyle
\setlength{\tabcolsep}{3.2pt}
\begin{tabularx}{\linewidth}{@{\hspace{\TableSidePad}}l l C C C@{\hspace{\TableSidePad}}}
\toprule
\textbf{Scenario} & \textbf{Method} & \textbf{Regret@25 (dB) $\downarrow$} & \textbf{F1 score $\uparrow$} & \textbf{Recall $\uparrow$} \\
\midrule
O1-28 & PMNet & 0.290 & 0.117 & 0.062 \\
\rowcolor[gray]{0.95}
O1-28 & Anchor-CKM & \textbf{0.170} & \textbf{0.788} & \textbf{0.982} \\
O1-60 & PMNet & 0.287 & 0.384 & 0.238 \\
\rowcolor[gray]{0.95}
O1-60 & Anchor-CKM & \textbf{0.214} & \textbf{0.879} & \textbf{0.931} \\
\bottomrule
\end{tabularx}
\TableNote{Low coverage: bottom $10\%$ of unobserved cells ($-115.50$/$-128.66$~dBm on O1-28/O1-60). Regret@25: true-best gap within the predicted top-$25$. Seed-$42$ O1 checkpoints.}
\end{minipage}
\end{table}

Table~\ref{tab:decision_metrics} shows that Anchor-CKM improves both decision-oriented metrics on O1-28/O1-60. The largest effect is in low-coverage detection: F1/recall increase from $0.117/0.062$ to $0.788/0.982$ on O1-28 and from $0.384/0.238$ to $0.879/0.931$ on O1-60. Regret@25 also decreases on both carriers, indicating that the predicted top-$25$ set is more likely to retain a strong fallback cell. The whole-tile accuracy improvement therefore appears together with better coverage-hole identification and candidate ordering, both relevant to beam fallback and handover protection.

\subsubsection{Boundary-Aware Error Analysis} We isolate LOS/NLOS transition behavior using an $8$\,m near-building band extracted from the O1 building footprints and pool valid unobserved band cells across transmitter-disjoint test tiles. Because only two tiles per carrier contain valid band nodes ($490$ on O1-60; $980$ pooled over O1-28/O1-60), this remains a localized error analysis, not a standalone benchmark.

\begin{figure*}[t]
\centering
\includegraphics[width=\textwidth]{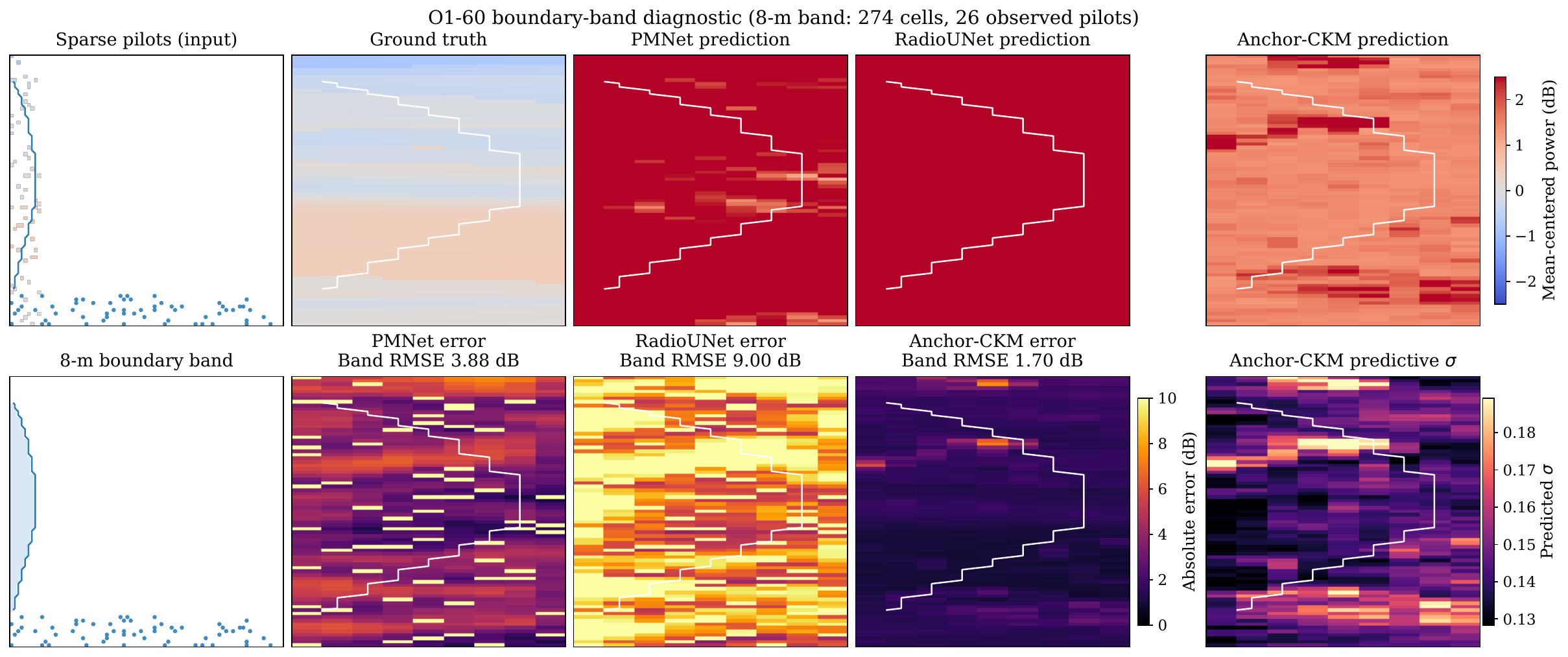}
\caption{Boundary-band error analysis on a representative O1-60 tile. Contour: 8-m near-building band; blue dots: pilots. Rows show mean-centered power, band/error maps, and Anchor-CKM predictive standard deviation; panel RMSE values are computed on the displayed tile, while pooled boundary-band RMSEs are reported in the text.}
\label{fig:near_boundary_diag}
\end{figure*}

Fig.~\ref{fig:near_boundary_diag} shows PMNet and RadioUNet smoothing the boundary contrast, while Anchor-CKM preserves a sharper transition and raises predictive standard deviation in the same region. On O1-60, boundary-band RMSE is $1.04$/$3.26$/$7.69$~dB for Anchor-CKM/PMNet/RadioUNet; removing the layout graph raises Anchor-CKM to $1.77$~dB. The pooled O1-28/O1-60 ordering is the same ($0.84$/$3.41$/$6.84$~dB). These results indicate a specific role for Stage~1: layout conditioning is not the main source of whole-field gain, but it helps localize regime transitions after the pilot-supported field has been formed. The predictive standard deviation is used here as a reliability cue around difficult boundaries; it is not evaluated as a separate edge detector.

\subsection{Ablation Studies}
\label{sec:ablation}

Table~\ref{tab:ablation} tests O1-60 ablations at $\rho=0.9$ under the full model's seed and geometry setting.

\begin{table}[!t]
\centering
\caption{Ablation on O1-60 at $\rho=0.9$.}
\label{tab:ablation}
\begin{minipage}{\SingleTableWidth}
\centering
\TableStyle
\begin{tabularx}{\linewidth}{@{\hspace{\TableSidePad}}Y c c@{\hspace{\TableSidePad}}}
\toprule
\textbf{Variant} & \textbf{RMSE (dB)} & \textbf{$\Delta$RMSE (dB)} \\
\midrule
\rowcolor[gray]{0.95}
\textbf{Anchor-CKM (Ours)} & \textbf{0.988} & \textbf{--} \\
\addlinespace[2pt]

\multicolumn{3}{l}{\textbf{Sparse measurements}} \\
\quad Without DtS & 1.622 & $+0.634$ \\
\addlinespace[2pt]

\multicolumn{3}{l}{\textbf{Operator/prediction head}} \\
\quad Without C-FNO & 1.414 & $+0.426$ \\
\quad Without Fourier coordinate features & 1.337 & $+0.349$ \\
\addlinespace[2pt]

\multicolumn{3}{l}{\textbf{Variance/layout}} \\
\quad MSE loss / no variance head & 1.271 & $+0.283$ \\
\quad With global FiLM & 1.143 & $+0.155$ \\
\quad Without layout graph & 1.389 & $+0.401$ \\
\addlinespace[2pt]

\multicolumn{3}{l}{\textbf{Layout graph}} \\
\quad Corner-only nodes & 1.102 & $+0.114$ \\
\quad Without edge types & 1.151 & $+0.163$ \\
\quad Without centroid links & 1.136 & $+0.148$ \\
\quad Without transmitter-visibility edges & 1.172 & $+0.184$ \\
\bottomrule
\end{tabularx}
\TableNote{$\Delta$RMSE is relative to full Anchor-CKM. Without DtS feeds $\mathbf{P}_{\mathrm{obs}}$ directly to C-FNO; MSE/no-variance uses MSE and removes the variance head. Same optimizer, schedule, masks, and seed.}
\end{minipage}
\end{table}

Removing DtS causes the largest degradation, increasing RMSE by $0.634$~dB; this identifies pilot-support stabilization as the bottleneck. Removing C-FNO or Fourier coordinate features is less damaging but still measurable, indicating that long-range propagation and continuous location-wise prediction matter after the support has been stabilized. The MSE/no-variance variant also degrades the mean estimate, suggesting that the heteroscedastic scoring rule acts as useful training regularization even though predictive variance is not part of the main accuracy objective. Layout variants have smaller effects than removing DtS, but they are consistent with the boundary study: spatial FiLM and graph structure mainly supply local correction near propagation-regime transitions. Overall, the O1-60 ablation is inconsistent with a purely layout-driven explanation of Table~\ref{tab:main}, because removing DtS is more harmful than removing any single layout component.

\subsection{Robustness and Deployment-Oriented Assessment}
\label{sec:robustness}

\subsubsection{Missing-Ratio Robustness} We evaluate Anchor-CKM, PMNet, RadioUNet, and RadioGAT at missing ratios $0.3$, $0.5$, $0.7$, $0.9$, and $0.95$. Following Remark~\ref{rem:random_sampling}, this random-probing sweep approximates opportunistic reporting, blockage, and trajectory-induced gaps; it is not a full mobility-trace model.

\begin{figure*}[t]
\centering
\subfloat[RMSE under varying missing ratios.\label{fig:robust_rmse}]{%
    \includegraphics[width=0.485\textwidth]{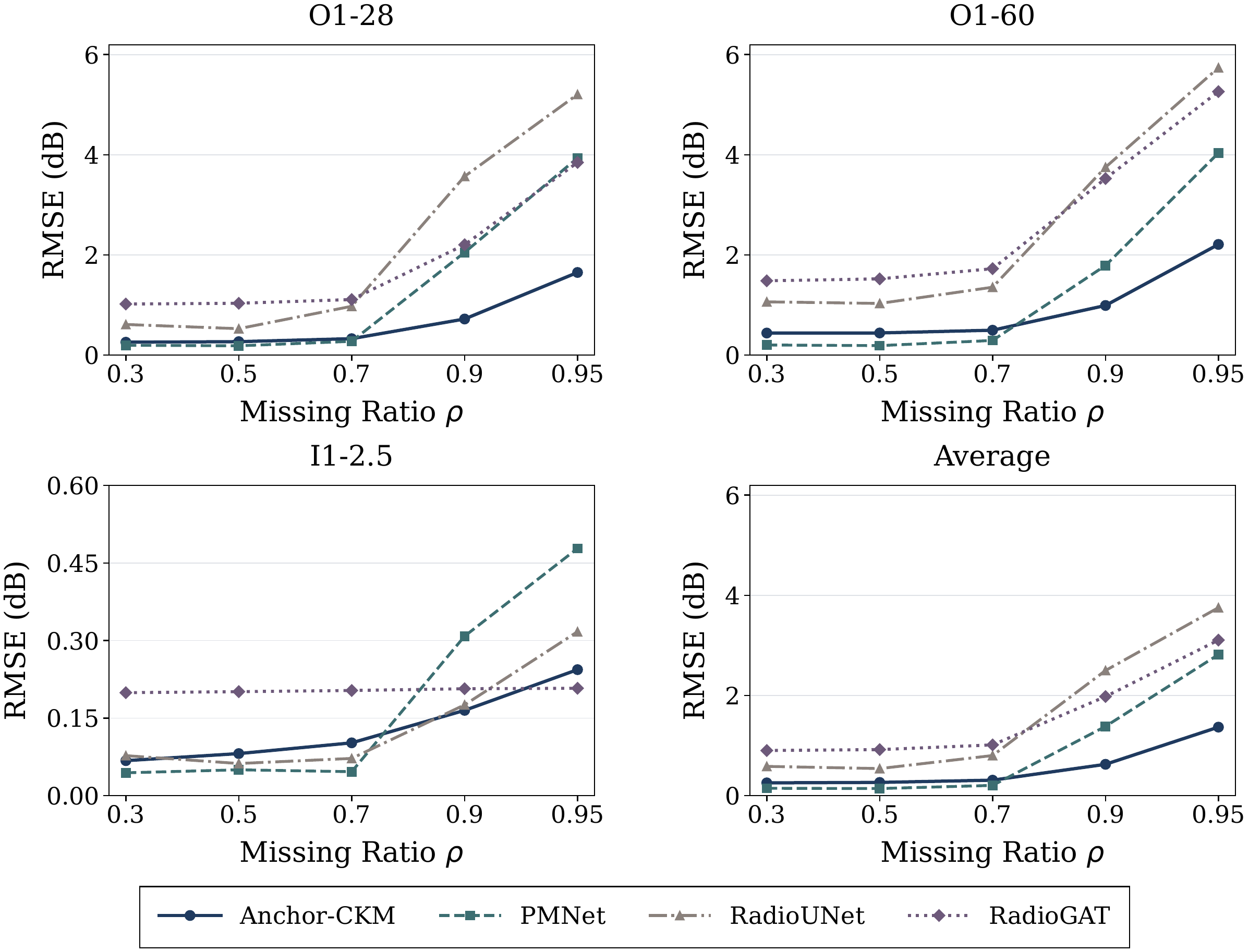}}
\hfil
\subfloat[P90 under varying missing ratios.\label{fig:robust_p90}]{%
    \includegraphics[width=0.485\textwidth]{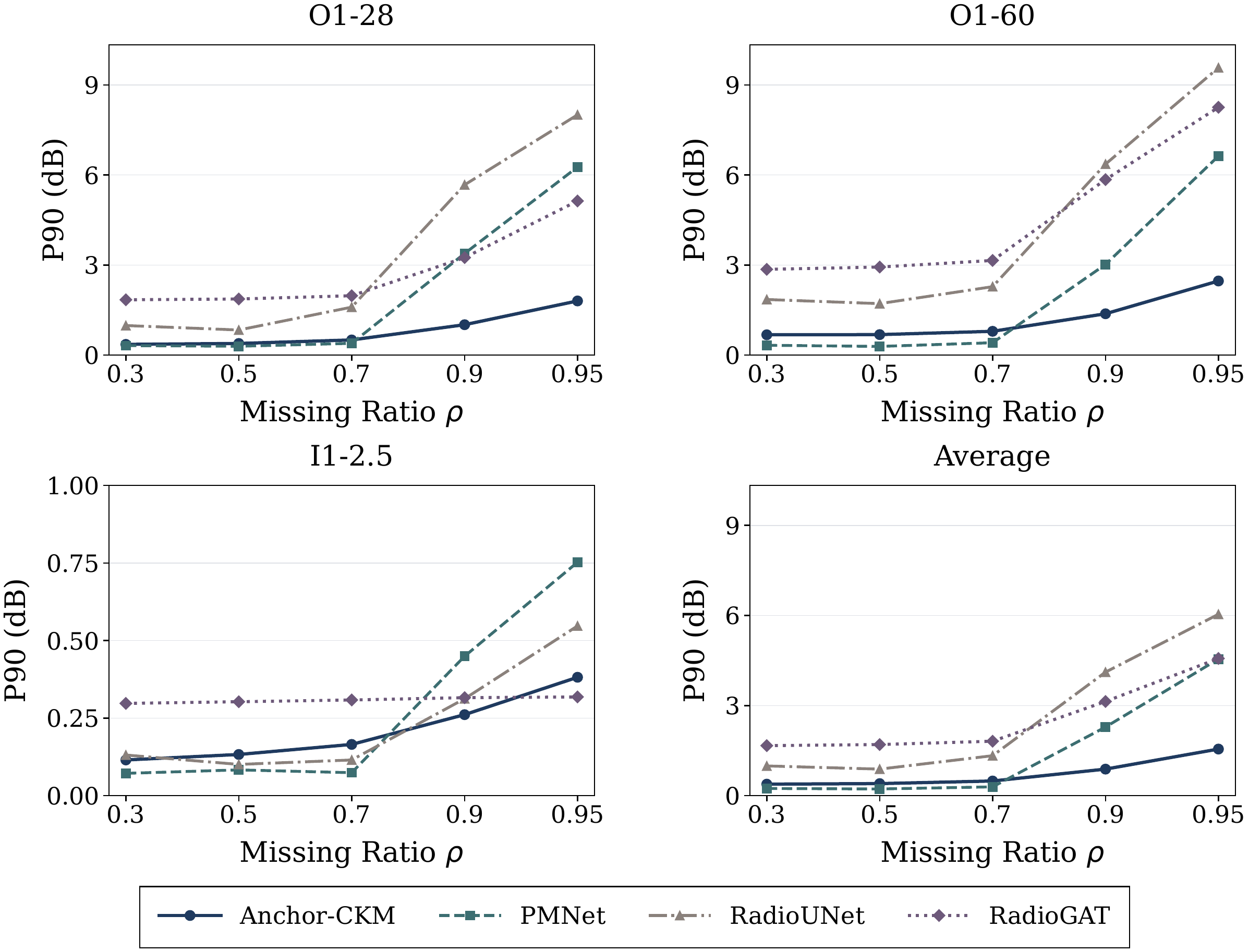}}
\caption{Missing-ratio robustness. RMSE and P90 are computed on unobserved valid cells for each scenario and the unweighted average.}
\label{fig:robustness_curves}
\end{figure*}

Fig.~\ref{fig:robustness_curves} shows that Anchor-CKM degrades more gradually than the reproduced deep baselines as the probing budget shrinks. At $\rho=0.95$, PMNet and RadioUNet exhibit sharply rising P90 errors, whereas Anchor-CKM's upper tail grows more slowly, consistent with DtS stabilizing the mask-distorted support before operator propagation. The P90 behavior is especially relevant for CKM refresh from sparse measurements because isolated high-error pockets can remain hidden when only average RMSE is monitored.

RadioGAT's nearly flat I1-2.5 curve should be interpreted cautiously because the weaker indirect prior may encode strong coordinate and topology regularity. We therefore interpret I1-2.5 primarily as a portability check under a weaker prior, not as evidence that all baselines remain measurement-anchored indoors.

\subsubsection{Predictive-Variance Assessment and Noisy Pilots} Table~\ref{tab:uncertainty_noise} summarizes deployment-oriented predictive-variance assessment under the fixed seed-42 O1 protocol. We report expected calibration error (ECE), $90\%$ interval coverage, and Corr$(|e|,\hat{\sigma})$ only as deployment-side indicators of the variance head, not as a full uncertainty-method comparison. On O1-28 and O1-60, ECE values are $0.056$ and $0.078$, and nominal $90\%$ intervals cover $93.7\%$ and $95.8\%$ of evaluation nodes. Together with Corr$(|e|,\hat{\sigma})$ values of $0.301$ and $0.137$, these quantities indicate useful but imperfect variance estimates; the O1-60 correlation is weaker, while the intervals remain mildly conservative without being overconfident.

\begin{table}[t]
\centering
\caption{Predictive-variance assessment and noisy-pilot robustness on O1 at $\rho=0.9$.}
\label{tab:uncertainty_noise}
\begin{minipage}{\SingleTableWidth}
\centering
\TableStyle
\setlength{\tabcolsep}{6pt}
\begin{tabularx}{\linewidth}{@{\hspace{\TableSidePad}}Y c c@{\hspace{\TableSidePad}}}
\toprule
\textbf{Metric} & \textbf{O1-28} & \textbf{O1-60} \\
\midrule
\multicolumn{3}{@{\hspace{\TableSidePad}}l}{Predictive-variance assessment} \\
\quad ECE $\downarrow$ & 0.056 & 0.078 \\
\quad 90\% interval coverage $\uparrow$ & 0.937 & 0.958 \\
\quad Corr$(|e|,\hat{\sigma})$ $\uparrow$ & 0.301 & 0.137 \\
\addlinespace[2pt]
\multicolumn{3}{@{\hspace{\TableSidePad}}l}{Noisy-pilot RMSE (dB) $\downarrow$} \\
\quad Anchor-CKM, $\sigma_{\mathrm{noise}}=0$ dB & 0.720 & 0.988 \\
\quad Anchor-CKM, $\sigma_{\mathrm{noise}}=3$ dB & 1.442 & 1.364 \\
\quad PMNet, $\sigma_{\mathrm{noise}}=3$ dB & 2.234 & 2.024 \\
\bottomrule
\end{tabularx}
\TableNote{ECE uses $15$ bins; coverage is nominal $90\%$; pilot noise is additive Gaussian in dB.}
\end{minipage}
\end{table}

Adding $3$-dB Gaussian pilot noise increases Anchor-CKM RMSE from $0.720$ to $1.442$~dB on O1-28 and from $0.988$ to $1.364$~dB on O1-60, but the margin over PMNet at the same noise level is preserved. Combined with the reported correlations, this supports using the predictive variance as a secondary deployment-side indicator rather than a primary accuracy mechanism.

\subsubsection{Multi-Target Reconstruction} To verify that Anchor-CKM is not restricted to power reconstruction, we evaluate a three-target CKM setting on O1-28 at $\rho=0.9$ with received power, AoA azimuth (AoA-Az), and AoA elevation (AoA-El). All methods are retrained for this setting, and Table~\ref{tab:multiparam} reports the standardized RMSE of each target together with their mean.

\begin{table}[t]
\centering
\caption{Three-target standardized RMSE on O1-28 at $\rho=0.9$.}
\label{tab:multiparam}
\begin{minipage}{0.99\columnwidth}
\centering
\TableStyle
\setlength{\tabcolsep}{3.6pt}
\begin{tabularx}{\linewidth}{@{\hspace{\TableSidePad}}l C C C C@{\hspace{\TableSidePad}}}
\toprule
\textbf{Method} & \textbf{Mean} & \textbf{Power} & \textbf{AoA-Az} & \textbf{AoA-El} \\
\midrule
LRMC & 0.985 & 0.930 & 1.012 & 1.012 \\
RadioUNet & 0.372 & 0.337 & 0.552 & 0.226 \\
PMNet & 0.270 & 0.230 & 0.390 & 0.192 \\
RadioGAT & 0.876 & 1.165 & 0.987 & 0.476 \\
RadioDiff & 2.041 & 2.179 & 2.055 & 1.888 \\
\rowcolor[gray]{0.95}
\textbf{Anchor-CKM} & \textbf{0.125} & \textbf{0.117} & \textbf{0.157} & \textbf{0.101} \\
\bottomrule
\end{tabularx}
\TableNote{Entries are standardized RMSE (lower is better), using each target's training-set mean and standard deviation.}
\end{minipage}
\end{table}

Anchor-CKM has the lowest mean standardized RMSE ($0.125$ versus PMNet's $0.270$) and improves every target, including AoA-Az and AoA-El. On O1-28, this indicates that the stabilized propagated field captures propagation structure useful for angular targets as well; transfer across more scenarios is left to future work.

\section{Conclusion}
\label{sec:conclusion}

We presented Anchor-CKM, a measurement-first, knowledge-aided framework for CKM reconstruction from sparse measurements and coarse environmental knowledge. In stringent low-coverage regimes, zero-filled pilots can form a mask-dominated input rather than a partially observed radio field: mask-induced zeros disrupt local neighborhoods and alias the mask spectrum with the channel spectrum. Anchor-CKM separates pilot-support stabilization, layout-conditioned spectral refinement, and continuous location-wise prediction, so environmental knowledge and spectral operators act only after a pilot-supported representation has been formed.

On the explicit-layout outdoor DeepMIMO scenarios, Anchor-CKM reduces received-power RMSE by $0.79$--$1.33$~dB relative to the strongest reproduced baseline. The indoor I1-2.5 result is best interpreted as a portability test under a weaker indirect prior, where Anchor-CKM obtains the lowest RMSE with a smaller margin. Ablations identify pilot-support stabilization as the largest observed contributor under the tested ablation setting, show that layout conditioning improves LOS/NLOS boundary reconstruction, and indicate that predictive variance provides a secondary reliability cue near coverage-hole and noisy-pilot regions. These observations support a measurement-first, knowledge-aided ordering for CKM maintenance from sparse measurements.

The current implementation still assumes uniform random probing and relies on quasi-static layout information. Future work should test whether the same ordering extends to trajectory-driven sampling, dynamic environments, multi-frequency or three-dimensional CKMs, and closed-loop resource management.

\bibliographystyle{IEEEtran}
\bibliography{refs}

\end{document}